\documentclass{amsart}       
\usepackage{graphicx}
\usepackage{mathptmx}

\usepackage{latexsym}
\usepackage{amsmath}
\usepackage{epsfig}
\usepackage{amssymb}
\usepackage{enumerate}
\usepackage{xcolor}

\usepackage{amsmath, amsfonts, amssymb} 
\usepackage{mathtools}
\usepackage{enumitem}

\usepackage[numbers, sort&compress]{natbib}
\usepackage[colorlinks]{hyperref}
\hypersetup{
         citecolor=blue,
    linkcolor=red,
}

\usepackage{color}
\usepackage{xcolor}

\newcommand{\change}{\textcolor{black}}
\newcommand{\orange}{\textcolor{black}}
\newtheorem{theorem}{Theorem}[section]

\newtheorem{corollary}[theorem]{Corollary}

\newtheorem{proposition}[theorem]{Proposition}

\usepackage{algorithm}
\usepackage[noend]{algpseudocode}

\begin{document}

\title{Autopolyploidy, allopolyploidy, and phylogenetic networks with horizontal arcs}	
%	\title{Ploidy profiles}
	%\thanks{}
	
	% Grants or other notes about the article that should go on the front
	% page should be placed within the \thanks{} command in the title
	% (and the %-sign in front of \thanks{} should be deleted)
	%
	% General acknowledgments should be placed at the end of the article.
	
%	\subtitle{}
	

\author{Katharina T. Huber \and Liam J. Maher}
         \address{School of Computing Sciences, University of East Anglia,
Norwich, UK}
   \email{K.Huber@uea.ac.uk}
         \address{School of Computing Sciences, University of East Anglia,
Norwich, UK }
  \email{L.Maher@uea.ac.uk}
% Second field is the short title of t he paper
% This should be shortened version of the title and no greater than 50 characters
	
	\date{\today}
	
\maketitle

\begin{abstract}
Polyploidization is an evolutionary process by which a species \orange{acquires multiple copies of} its complete set of chromosomes. The reticulate nature of the signal left behind by it means that phylogenetic networks offer themselves as a framework to reconstruct the evolutionary past of species affected by it. The main strategy for doing this is to first construct a so called multiple-labelled tree and to then somehow derive such a network from it. \change{The following question therefore arises}: How much can be said about that past if such a tree is not readily available? By viewing a polyploid dataset as a certain vector which we call a ploidy (level) profile we show that, among other results, there always exists a phylogenetic network in the form of a beaded phylogenetic tree with additional arcs that realizes a given ploidy profile. Intriguingly, the two end vertices of almost all of these additional arcs can be interpreted as having co-existed in time thereby adding biological realism to our network, a feature that is, in general, not enjoyed by phylogenetic networks.  In addition, we show that our network may be viewed as a generator of ploidy profile space, a novel concept similar to phylogenetic tree space that we introduce to be able to compare phylogenetic networks that realize one and the same ploidy profile. We illustrate our findings in terms of a publicly available Viola dataset.
\end{abstract}

%\keywords{phylogenetic network \and ploidy profile \and polyploid phylogenetics \and ploidy profile space}
%\MSC 05C05\sep 05C20 \sep 05C85 \sep 05D15\sep 92D15

%	\keywords{Phylogenetic network \and ploidy profile \and polyploid phylogenetics \and \and ploidy space}

\section{Introduction}
Polyploidization is an evolutionary phenomenon thought to be one of the key players in plant evolution.
It has, however, also been observed in fish \cite{LI03}, and 
fungi \cite{AM12} and
arises when a species \orange{acquires multiple copies of} its full set of chromosomes. \orange{This can be the result of, 
for example,  a species undergoing whole genome duplication (autopolyploidization)  or through acquisition} of a further complete set of chromosomes via interbreeding with a different, usually closely related, species
(allopolyploidization) \orange{\cite{AM12}} (see also \cite{DS-B17} \orange{who point out that the definitions of allopolyploidy and autopolyploidy are controversial}). 
%Autopolyploids have been found to often be viable but less fertile than their parental species \cite{JS02}. 
 Examples of autopolyploids include crop potato
\cite{TPSC11} and bananas and watermelon \cite{VBDG00}
%have been found to often be viable but less fertile than their parental species \cite{JS02}. Examples of autopolyploids include plants in the genus {\em Tolmiea} \cite{SR86}. 
%Allopolyploids on the other hand have been found to be more likely to be fully fertile, 
and examples of allopolyploids include bread wheat \cite{M-et-al14} and oilseed rape. Understanding better how polyploids have arisen (and still arise) therefore has potentially far reaching consequences.

Many tools for shedding light into the evolutionary past of a polyploid data set such as PADRE \cite{LSHM2009} and AlloPPnet \cite{JSO13} start with a multiple-labelled tree, 
sometimes also called a MUL-tree or a multi-labelled tree.
\begin{figure}[h!]
	\centering
		\includegraphics[scale=0.2]{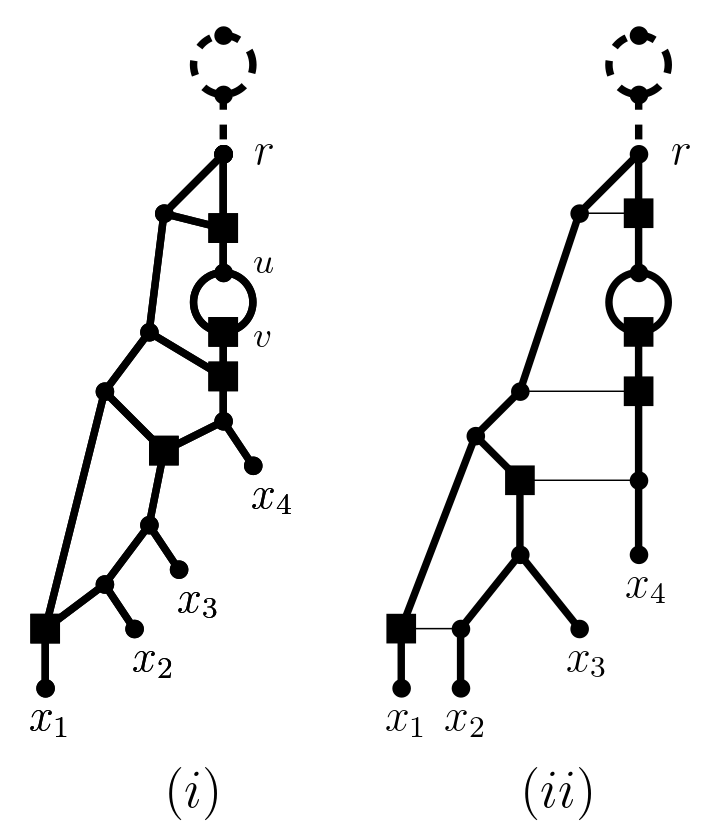}
	\caption{\label{fig:phy-network}
				(i) One of potentially many phylogenetic networks that realize
		the ploidy levels $14, 12, 12, 10$ 
		of a set $X=\{x_1,x_2,x_3,x_4\}$ of taxa where $14$ is the ploidy level of $x_1$, the ploidy level of $x_2$ and $x_3$
		is $12$, respectively and the ploidy level of $x_4$ is $10$. 
		To improve clarity of exposition, we always assume that unless indicated otherwise, arcs are directed away from the root
		(which is always at the top). 
		(ii) The network in (i) represented  in such a way that 
		every reticulation vertex (indicated throughout the paper
		by a square and defined below)
		has precisely one incoming horizontal arc implying that
		the end vertices of such an arc represent ancestral species that have existed at the same point in time. 
		In both (i) and (ii) the phylogenetic network resulting from deleting the dashed bead and its dashed outgoing arc
		realizes the ploidy profile $\vec{m}=(7, 6, 6, 5)$. 
		}
\end{figure}
These types of trees differ from the standard phylogenetic trees by allowing two or more leaves to be labelled with the same species. In the case of PADRE a (phylogenetic) network is then produced from such a tree by folding it up as described in, for example, \cite{HOLM06}. Referring the interested reader to Figure~\ref{fig:phy-network}(i) for an example and below for definitions, it suffices to
say at this stage that a phylogenetic network is a directed graph with leaf set a set of taxa \orange{(e.g. species)} of interest,
a single root (usually drawn at the top), and no directed cycles. Note that to be able to account for autopolyploidy, we
deviate from the standard definition of a phylogenetic network
(see e.\,g.\,\cite{S16}) by also allowing it to contain {\em beads}, that is, \orange{pairs of} parallel
arcs, as is the case in the networks depicted in Figure~\ref{fig:phy-network}.
Polyploidization events  are represented in such networks as {\em reticulation vertices}, that is, vertices with more
than one arc coming into them. For clarity of exposition, 
we indicate \orange{reticulation vertices}  throughout the paper in terms of squares.
Although PADRE is generally fast and not constrained by an upper limit on the ploidy levels in a dataset of interest, its underlying assumptions  imply that it is highly susceptible to noise in the multiple-labelled tree from which
\orange{the network} is constructed. In the case of AlloPPnet, a phylogenetic network is inferred using, among other techniques, the multispecies coalescent to account for incomplete lineage sorting.
The computational demands of AlloPPnet however mean that it can only be applied on relatively small data sets that contain only diploids and tetraploids \cite{R21}.

\orange{One approach to obtain an input  multiple-labelled tree for PADRE is to try and  construct it as a consensus multiple-labelled tree from a set of multiple-labelled gene trees. This task is relatively straight-forward for phylogenetic trees by applying, for example, some kind of consensus approach to the collection of clusters induced by the trees. The corresponding approach for constructing a consensus multiple-labelled tree from a collection of multiple-labelled gene trees however gives rise to a computationally hard decision problem \cite{HLMS09}. A consensus multiple-labelled tree might therefore not always be readily available for a dataset.}
%The construction of a biologically meaningful multiple-labelled tree that somehow combines the evolutionary 
%scenarios that gave rise to a polyploid dataset is however in general not easy \cite{HLMS09}. 
The \orange{following} question therefore \orange{arises}: How much can we say about the \orange{reticulate} evolutionary past of a polyploid dataset if a multiple-labelled tree is not readily  available?  \orange{Since one of the signatures left by polyploidization is the {\em ploidy level} of a species (i.\,e.\,the number of copies of the complete set of chromosomes of that species),}
we address this 
question in terms of a dataset's ploidy levels or more precisely the ploidy levels of the taxa that make up the dataset \orange{using phylogenetic networks as a framework. Interpreting the ploidy level of a \orange{species} as the number of directed paths from the root of a phylogenetic network $N$ to the leaf in  $N$ that represents that \orange{species},  Figure~\ref{fig:phy-network}(i) implies that, in general, ploidy levels do not preserve the topology of the phylogenetic network that induced them. For example,  swapping $x_2$ with $x_3$ in that network results  in a phylogenetic network that induces the same ploidy levels on $\{x_1,\ldots, x_4\}$ as the network pictured in Figure~\ref{fig:phy-network}(i). We are therefore interested in understanding to what extent a phylogenetic network representing the evolutionary past 
	of a polyploid dataset can be derived solely from the ploidy levels of the} species that make up the dataset.
%in the sense that the number of directed paths from the root of $N$ to a leaf $x$ equals the ploidy level of $x$. 

Note that since polyploidization events are assumed to be rare, we are particularly interested in phylogenetic networks that enjoy this property and also aim to minimize the number of reticulation vertices. From the perspective of reducing the complexity of our mathematical arguments this immediately implies that we may assume the ploidy level of a taxon to \orange{not be even}. Indeed, if we have a dataset where every ploidy level is of the form $m=2m'$, some positive integer $m'$, \orange{then since polyploidization events are assumed to be  rare}, we may assume the last common ancestor of the dataset's taxa to have undergone autopolyploidization. The root of a phylogenetic network $N$ that represents the evolutionary past of the dataset is therefore contained in a bead and that bead accounts for the factor $2$ in $m$. Thus, the phylogenetic network 
obtained by removing this bead and the arc that joins it to the rest of $N$ is a phylogenetic network that represents the factor $m'$ of $m$ in terms of numbers of directed paths from the root to the leaves. 

In view of the above, we call any (finite) vector of positive integers that is indexed by a (finite, non-empty) set $X$  a {\em ploidy profile (on $X$)}.  Although related to the recently introduced {\em ancestral profiles} \cite{ESS19} (but also see \cite{BESS21}) ploidy profiles differ from them by only recording the number of directed paths from the root of a phylogenetic network $N$ to every leaf of $N$. \orange{Ancestral} profiles \orange{on the other hand} record the number of directed paths from every non-leaf vertex in 
$N$ to all the leaves below that vertex. In particular, a ploidy profile is an element of an ancestral profile of a phylogenetic network.

To help motivate our approach for addressing our question, consider the phylogenetic network $N$ with leaf set $X=\{x_1,x_2,x_3,x_4\}$ depicted in Figure~\ref{fig:phy-network}(i) where the square vertices at the end of each pair of two parallel arcs represent autopolyploidization and the remaining four square vertices represent allopolyploidization. Then taking for each taxon
$x$ in $X$ the number of directed paths from the root of the network to $x$ results in the ploidy profile $\vec{m}=(14,12,12,10)$ where 
the first component is indexed by $x_1$, the second by $x_2$ and so on. Each component in $\vec{m}$ is of the form $2m$, some positive integer $m$, and the phylogenetic network rooted at $r$ obtained by removing the dashed bead together with the dashed arc coming into $r$ represents the ploidy profile $\vec{m'}=(7,6,6,5)$ in terms of numbers of directed paths from $r$ to the leaves. With this in mind, we say that
 a ploidy profile $\vec{m}=(m_1,\ldots, m_n)$, $n\geq 1$ on  $X=\{x_1,\ldots, x_n\}$ is 
{\em realized} by a phylogenetic network $N$ with leaf set $X$
if, for all $1\leq i\leq n$, the number of directed paths from
the root of $N$ to $x_i$ is $m_i$. For example, both phylogenetic networks pictured in Figure~\ref{fig:phy-network} realize the ploidy profile $\vec{m} = (14,12,12,10)$ indexed by $X = \{x_1,x_2,x_3,x_4\}$. 
%, phrased within our framework, they do not contain not only the ploidy levels in a dataset but also information about how they have increased over time. 

\orange{Contributing to the emerging field of {\em Polyploid Phylogenetics} \cite{R21},
a first inroad into our question was made in \cite{HM21}}
by studying the hybrid number of a ploidy profile $\vec{m}$, that
is, the minimal number of polyploidization events required by a phylogenetic network to realize $\vec{m}$. 
As it turns out, the arguments underlying the results in \cite{HM21} 
largely rely on a certain iteratively constructed network
that realizes $\vec{m}$. Denoting for a choice \orange{$C$} of initializing network the generated network by \orange{$N(\vec{m})=N_C(\vec{m})$} and changing the network initializing that construction in a way that does not affect the main findings in \cite{HM21} (see below for details), 
%and denoting it by $N(\vec{m})$, 
we show that even more can be said about ploidy profiles.  For example, our first result (Proposition~\ref{prop:connected}) shows that 
$N(\vec{m})$ may be thought of as a generator of ploidy profile space (defined in a similar way as phylogenetic tree space) in the sense that any realization of $\vec{m}$ can be reached from 
 $N(\vec{m})$ via a number of multiple-labelled tree editing operations and reticulation vertex splitting operations. As an immediate consequence of this we obtain  
a distance measure for phylogenetic networks that realize one and the same ploidy profile.
On a more speculative level it might be interesting to
see if  $N(\vec{m})$ lends itself as a useful prior for a
Bayesian method along the lines as described in \cite{IJJM21}.

Our second result (Theorem~\ref{theo:specialcase}) shows that a key concept introduced in \cite{HM21} called the simplification sequence of a ploidy profile $\vec{m}$ is in fact closely related to the notion of a cherry reduction sequence \cite{ESS19} for $N(\vec{m})$, also called a cherry picking sequence in 
\cite{JM21}. In case autopolyploidy is not suspected to have played a role in the evolution of a dataset, this implies that the network $N(\vec{m})$ can also be reconstructed from phylogenetic networks on three leaves called trinets \cite{ST21}. These can be obtained from a dataset using, for example, the TriLoNet software \cite{OWvIM16}.

Exemplified in terms of the phylogenetic network depicted in
Figure~\ref{fig:phy-network}(ii) for the ploidy profile $(14,12,12,10)$, our third result (Theorem~\ref{theo:treebased}) implies that
for any ploidy profile we can always find a phylogenetic network realizing it in the form of a phylogenetic tree that potentially contains  beads
to which additional arcs have been added and at most one of those arcs is not horizontal. 
In the context of this it is important to note that, in general, a phylogenetic network cannot be thought of as a phylogenetic tree with additional arcs let alone horizontal ones. The reason for this is that horizontal arcs imply that the ancestral taxa joined by such an arc must have existed at the same time (see also \cite[Section 10.3.3]{S16} for more on this and the Viola dataset below for an example).

 The remainder of the paper is organized as follows. In the next section, we review relevant basic terminology surrounding graphs, phylogenetic networks and ploidy profiles. For a ploidy profile $\vec{m}$, we outline the construction of the network $N(\vec{m})$ in \orange{Section~\ref{sec:realization}}. 
This includes the definition of the simplification sequence for $\vec{m}$.
Subsequent to this, we  introduce ploidy profile space 
\orange{in Section~\ref{sec:comparison}} and \orange{also} establish Proposition~\ref{prop:connected}
\orange{in that section}.
Sections~\ref{sec:weakly-orchard-and-beyond} is concerned with establishing
%use the concepts of a HGT-consistent labelling to show that $N(\vec{m})$ 
%can in fact be thought of as a phylogenetic tree with beads 
%with additional arcs such that at most one of the additional arcs is not a horizontal arc (
Theorems~\ref{theo:specialcase} and \ref{theo:treebased}. To do this, we use Theorem~\ref{theo:arc-rich} \orange{which we establish in Section~\ref{sec:weak-hgt-consistent}. That theorem}
is underpinned by the concept of a so called HGT-consistent labelling introduced in \cite{IJJMZ20},
a notion that we extend to our types of phylogenetic networks here. 
%. As part of this we also 
%clarify the relationship between the simplification sequence for a ploidy profile $\vec{m}$ and a complete cherry reduction sequence \cite{ESS19} for $N(\vec{m})$ 
%establish 
%Theorem~\ref{theo:specialcase}. 
 In the last but one section, we  employ a simplified version of a Viola dataset from \cite{MJDBBBO12} to help explain our 
 findings within the context of  a real biological dataset.
 We conclude with some potential directions of further research in the last section.

\section{Preliminaries} 
 \label{sec: preliminaries}
 We start with introducing basic concepts surrounding 
 phylogenetic networks. 
 %and then briefly describe two 
 %basic operations concerning phylogenetic networks.
 %For the convenience of the reader, we illustrate them in Figures~\ref{fig:fun-construct} and \ref{fig:intro-fig} 
 %by means of an example.
 Throughout the paper, we assume that $X$ is a (finite) set
 that contains at least one element. Also, we denote the number of elements in $X$ by $n$.

 \subsection{Graphs}
 \label{sec:graphs}
 Suppose for the following that $G$ is a directed acyclic graph
 with a single root which might
 contain parallel arcs but no loops. 
 %Then we denote the vertex set of $G$ by $V(G)$ and 
 %its set of arcs by $A(G)$. 
 We denote an arc starting at a vertex $u$ and ending in a vertex $v$ by $(u,v)$. If there exist \change{two} arcs from $u$ to $v$ then we refer to the pair of arcs from $u$ to $v$ as a {\em bead} of $G$.

 Suppose $v$ is a vertex of $G$. Then  
 we refer to the number of arcs coming into $v$
 as the {\em indegree} of $v$ in $G$ and denote it by $indeg(v)$.
 Similarly, we call the number of outgoing arcs of $v$ the
 {\em outdegree} of $v$ in $G$ and denote it by $outdeg(v)$. We call $v$ the {\em root} of $G$, if $indeg(v)=0$,  and we 
 call $v$ a {\em leaf} 
 of $G$ if $indeg(v)=1$ and $outdeg(v)=0$. We denote 
 the set of vertices of $G$ by $V(G)$ and the set of leaves of $G$ by $L(G)$. 
 We call $v$
 a {\em tree vertex} if $outdeg(v)=2$ and $indeg(v)=1$, and
 we call $v$ a {\em reticulation vertex} if $indeg(v)=2$ and 
 $outdeg(v)=1$. 
 %We denote the set of reticulation vertices of $G$ by $H(G)$.
 If $w$ is also a vertex in $G$ then we say that $w$
 is {\em below} $v$ if either $v=w$ or there exists a directed 
 path from the root of $G$ to $w$ that crosses $v$. If $w$ is
 below $v$ and $v\not=w$ then we say that $w$ is 
 {\em strictly below} $v$. A {\em parent} of a vertex $v$ is the vertex connected to $v$ on the path to the root. A {\em child} of a vertex $v$ is a vertex of which $v$ is the parent.

Suppose $a$ and $b$ are two distinct leaves of $G$. Then 
the set $\{a,b\}$ is called a {\em cherry} of $G$ if the parent $p_a$ of $a$ is also the parent of $b$. If the parent $p_b$ of $b$ is a reticulation vertex
  and there is an arc $(p_a,p_b)$ from $p_a$ to $p_b$  then  the
  set $\{a,b\}$ is called a {\em reticulate cherry}. In this case, the arc $(p_a,p_b)$ is called a {\em reticulation arc}
  of $G$ and the leaf $b$ is called a {\em reticulation leaf} of $G$. 
  
  For example, $x_1$ is the reticulate leaf 
  of the reticulate cherry $\{x_1,x_2\}$
  in the graph depicted in Figure~\ref{fig:phy-network}(i). The parent of $x_2$ is a tree vertex and the parent of $x_1$ is a reticulation vertex. The vertices $u$ and $v$ form a bead.
  
 \subsection{Phylogenetic networks and trees}
 Suppose $G$ is a graph as described above.
 %in Section~\ref{sec:graphs}. 
 If $G$ contains at least three vertices then we call $G$ a {\em (phylogenetic) network (on $X$)}
 if the outdegree of the root $\rho$ of $G$ is 2, the leaf set of $G$ is $X$, and every vertex other than 
 $\rho$ or a leaf is a tree vertex or a reticulation vertex. 
  Note that  our definition of a phylogenetic network differs from the standard definition of such an object (see e.g. \cite{S16}) by allowing the network to contain beads and $X$ to contain a single element.  To distinguish between our type of phylogenetic networks and the
 standard type of phylogenetic networks we refer to the
 latter as {\em beadless} phylogenetic networks.
 %Suppose $G$ is a phylogenetic 
 %network on $X$. Then following \cite{BS07},
 %we define the {\em reticulation number} $h(G)$ of $G$
% to be 
% $$
% h(G)=\sum_{h\in H(G)} (indeg(h)-1.).
% $$ 
% Note that that number is sometimes also called the
% hybrid number of $G$ \cite{BS07}. 
We call a phylogenetic network 
 (on $X$) a {\em phylogenetic tree (on $X$)} if it does not contain any reticulation vertices.

 %Suppose that $N$ is a 
 %phylogenetic network on $X$. 
 %Then we denote the number of
 %directed paths form the root $\rho_N$ of $N$ to $x$
 %by $m_N(x)$. In case $N$ is clear from the context,
 %we will write $m(x)$ rather than  $m_N(x)$. 
 %For $N'$
 %a further phylogenetic network on $X$ we say that
 %$N$ and $N'$ are {\em equivalent} if there exist a
 %graph isomorphism between $N$ and $N'$ that is the
 %identity on $X$. Furthermore, we say that $N'$ is
 %a {\em resolution} of $N$ if $N'$ obtained from $N$ by %resolving all vertices in $H(N)$ so that every 
 %vertex in $H(N')$ has indegree two. Note that for any %resolution $N'$ of $N$, we
 %have $h(N)=|H(N')|$. 
 %
 %Following
 %\cite{IJJMZ19}, we call an induced subgraph $N'$ 
 %of a phylogenetic network $N$ 
 %with $|V(N')|=2=|A(N')|$ a {\em bead} of $N$. 
 Finally, we call a phylogenetic
 network $N$ on $X$ such that $N$ is either a phylogenetic tree on $X$ or
 every reticulation vertex of $N$ is contained in a bead
 a {\em beaded tree} (see e.g. \cite{IJJMZ19}
 and \cite{HLM21} for more on such graphs).
 
 \subsection{Ploidy profiles}
 Let $X=\{x_1,\ldots, x_n\}$. Then, as mentioned in the introduction,
  a {\em ploidy profile $\vec{m}=(m_1,\ldots, m_n)$ (on $X$)} is an 
 $n$-dimensional vector of positive integers such that each
 component is indexed by an element in $X$. For ease of readability, we will 
 assume from now on that the elements in $X$ are always ordered in such a 
 way that $x_i$ indexes component $m_i$ of $\vec{m}$, for all $ 1\leq i\leq n$,
 and that $\vec{m}$ is in {\em descending order},
 that is,  $m_i\geq m_{i+1}$ holds for all $1\leq i\leq n-1$. For example the vector $\vec{m}=(7,6,6,5)$
 is a ploidy profile on $X=\{x_1,x_2.x_3,x_4\}$ 
 where $x_1$ indexes the first component i.\,e.\,$7$, $x_2$ indexes the second component, and so on.
 
 Suppose $\vec{m}=(m_1,\ldots, m_n)$ is a ploidy profile
 on $X$. If $m_1\geq 2$ and  all other components of $\vec{m}$ are 1 then we call $\vec{m}$
a {\em simple} ploidy profile. If $\vec{m}$ is a  simple
ploidy profile and  $n=1$ then we call $\vec{m}$ a {\em strictly simple} ploidy profile. Finally, we say that
a phylogenetic network is a {\em realization} of $\vec{m}$
if it realizes $\vec{m}$ (as defined in the introduction).
%Finally, we say that a phylogenetic
% network $N$ on $X$ {\em realizes}
% $\vec{m}$ if $m_i$ is the number of directed paths from the root %of $N$ to $x_i$, for all $1\leq i\leq n$. 
For example,
 the ploidy profile $\vec{m}=(77,1,1,1)$ is simple but not strictly simple and the
 ploidy profile $(77)$ is strictly simple. The phylogenetic networks depicted
 in Figure~\ref{fig:phy-network} are realizations of the ploidy profile $\vec{m} = (14,12,12,10)$. 
 %We say that $N$ {\em attains} $h(\vec{m})$ if $N$ realizes
 %$\vec{m}$ and $ h(\vec{m})=h(N)$. In that case, we refer %to $N$
 %as an {\em attainment} of $\vec{m}$.
 %
%For the convenience of the reader, we depict for the simple ploidy profile $\vec{m_1}=(5,1,1)$ on $X_1=\{x_1,x_2,x_3\}$ an attainment of $\vec{m_1}$ in Figure~\ref{fig:m2=1}. 

\section{Realizing ploidy profiles}
\label{sec:realization}

We start this section by introducing further terminology which
will allow us to construct our network $N(\vec{m})$ from a ploidy profile $\vec{m}$. To avoid undesirable uniqueness issues, we remark that our construction is slightly different from the construction
of the corresponding network for $\vec{m}$ introduced in \cite{HM21} in that
we choose a different network with which we initialize its construction.  As part of our construction we also include a worked example at the end of this section.
 
Suppose $\vec{m}=(m_1, m_2, \ldots, m_n)$ is a ploidy profile on $X=\{x_1, x_2, \ldots, x_n\}$. Then we first construct 
a sequence $\sigma(\vec{m})$ of ploidy profiles from $\vec{m}$
which we call the {\em simplification sequence} for $\vec{m}$.
This sequence starts with the ploidy profile $\vec{m}$ and terminates with a certain simple ploidy profile $\vec{m_t}$ which
we call the {\em terminal element} of $\sigma(\vec{m})$. \orange{If $\vec{m}$ is simple,
	then we define  $\sigma(\vec{m})$ to contain only $\vec{m}$. Thus, $\vec{m}= \vec{m_t}$ in this case}. 

Assume for the following that $\vec{m}$ is not simple. \orange{To define 
$\sigma(\vec{m})$ in this case, assume furthermore}
that all ploidy profiles in $\sigma(\vec{m})$
have been constructed already up to and including a ploidy profile
$\vec{m'}=(m_1',\ldots, m_q')$ on some set $X'=\{x_1',\ldots, x_q'\}$, some $q\geq 1$. \orange{If $\vec{m'}$ is simple then we define $\vec{m'}$ to be 
$\vec{m_t}$. So assume that $\vec{m'}\not=\vec{m_t}$.}  
Put $\alpha = m_1' - m_2'$. Let $X''$ denote the set that
indexes  the next ploidy profile 
in $\sigma(\vec{m})$ which we call $\vec{m''}$.
Then $\vec{m''}$ \orange{and $X''$ are} obtained from $\vec{m'}$ 
\orange{and $X'$} by applying one of the following cases. 
\begin{enumerate}

\item[$\bullet$] If  $\alpha=0$ then delete $m_1'$ from
$\vec{m'}$ and its index $x_1'$ from $X'$. To obtain $X''$, rename the elements $x_{k}'$ as  $x_{k-1}''$, $2\leq k\leq q$.

\item[$\bullet$] If $\alpha>m_2'$ then replace $m_1'$ by $\alpha$. The set $X''$ is $X'$ in this case 
and the indexing of the components of $\vec{m''}$ is as in $\vec{m'}$. 

\item[$\bullet$] If $\alpha\leq m_2'$ then 
\orange{remove $m'_1$ from $\vec{m'}$ and its index $x_1'$ from $X'$ to obtain a ploidy profile $\vec{a}$ on  $X'-\{x'_1\}$. Into $\vec{a}$, insert  $\alpha$ so that the resulting integer vector $\vec{b}$ is a ploidy profile
on $X''=X'-\{x'_1\}\cup \{x\} $ where $x$ is an element not already contained in $X'$. That element is used to index $\alpha$  in $\vec{b}$. As $\vec{a}$ might already contain a component with value $\alpha$, we also require that $\alpha$ is inserted into $\vec{a}$ directly after the last occurrence of $\alpha$  to ensure that $\vec{b}$ is unique. Next, relabel the elements in $X''$ so that the indexing of
$\vec{b}$
conforms to our indexing convention for ploidy profiles. Finally, put $\vec{m''}=\vec{b}$.}
\end{enumerate}
This completes the construction of the simplification sequence for $\vec{m}$. \change{To aid intuition, we present  the simplification sequence for  the ploidy profile $\vec{m}=(7,6,6,5)$ at the end of this section.}

To obtain our realization for our ploidy profile $\vec{m}$,
 we next choose a {\em core network} for $\vec{m}$, that is, a 
 phylogenetic network $N$ that realizes the terminal element
 $\vec{m_t}$. This is always possible since
 for any simple ploidy profile $(m_1,\ldots, m_n)$ on $X=\{x_1,\ldots, x_n\}$
 such a network can be obtained using the following naive approach.  Take a directed path $P$ with $n+2(m_1-1)$ vertices 
 and label the first $m_1-1$ vertices on $P$ by $v_i$, $1\leq i\leq m_1-1$ and the
 next $n-1$ vertices by $w_i$, $2\leq i\leq n$. Starting at the other end of $P$, label the first vertex $x_1$ and
 the remaining $m_1-1$ vertices by $v_i'$, $1\leq i\leq m_1-1$.
 Finally, for all $2\leq i\leq n$ attach the arc $(w_i,x_i)$
 and, for all $1\leq i\leq m_1-1$, the arc $(v_i,v_i')$. By construction, the resulting
 graph is a phylogenetic network \orange{(without horizontal arcs)} that realizes $\vec{m_t}$. To
 keep the description of the construction of 
 $N(\vec{m})$ from $N(\vec{m_t})$ compact, we refer the
 interested reader to \orange{Section~\ref{sec:weak-hgt-consistent}} 
 %Section~\ref{sec:weak-hgt-consistent} 
 for the construction of a more attractive  
 choice of core network for $\vec{m}$.
 
 Starting with a core network $N$ for $\vec{m}$ we then apply a traceback through $\sigma(\vec{m})$ to obtain $N(\vec{m})$
 via the addition of  vertices and arcs (see e.\,g.\,\cite{JM21},
 where, in a different context, this process was called ``adding'' vertices). For
 this we distinguish the same cases as before. More precisely, if $\vec{m}$ is simple and therefore
 the terminal element of $\sigma(\vec{m})$, we define $N(\vec{m})$
 to be $N$. So assume that $\vec{m}$ is not simple and that, \orange{starting at
 	$\vec{m_t}$, for}
 all ploidy profiles in $\sigma(\vec{m})$ up to and including
 a ploidy profile $\vec{m''}$ on $X''$ a realization of them
 has already been constructed. Let $N''$ denote the 
 realization obtained for $\vec{m''}$. As before, let $\vec{m'}$ on
 $X'$ denote the ploidy profile in $\sigma(\vec{m})$ that precedes
 $\vec{m''}$. For clarity of presentation of the main ideas, we
 remark that in each of the following cases the set $X'$ is obtained
 from $X''$ by reversing the indexing that formed part of the
 corresponding case in the construction of $\sigma(\vec{m})$.
 
	\begin{enumerate}
		\item[$\bullet$] If  $\alpha=0$ then replace $x_1''$ by the cherry $\{x_1',x_2'\}$.
		\item[$\bullet$] If $\alpha > m_2$, then subdivide the incoming arc of
	$x_1''$ by a new vertex $u$. Next, subdivide the incoming arc of $x_2''$ by a
	vertex $v$ and add the new arc $(v,u)$.
	\item[$\bullet$] If $\alpha \leq m_2$ \orange{then let $\vec{a}$ and $\vec{b}$ as in the corresponding case of the construction of $\sigma(\vec{m})$. Let $j$ denote the index of the 
	component of $\vec{b}$ that was inserted into  $\vec{a}$ as part of the construction of  $\vec{b}$}.
	 Then subdivide the incoming arc of $x_{j}''$ by a new vertex $v$ and replace $x_1''$ by the cherry $\{x_1',x_2'\}$. Next, subdivide the incoming arc of $x_1'$ by a new vertex $u$. Finally, add the arc $(v,u)$ and delete $x_j''$ and its incoming arc $(u,x_j'')$ (suppressing the resulting indegree and outdegree one vertex).
	
\end{enumerate}

To illustrate the construction of $N(\vec{m})$, consider the ploidy profile $\vec{m} = (7,6,6,5)$ on $X=
\{x_1,x_2,x_3,x_4\}$. Then the sequence $\vec{m}$, $(6,6,5,1)$, $(6,5,1)$ $(5,1,1)$ is the
simplification sequence $\sigma(\vec{m})$ for $\vec{m}$. Since $\vec{m_t} = (5,1,1)$, the
phylogenetic network  depicted on the left of Figure~\ref{fig:simpli-traceback-cherry-picking-only} is
a core network for $\vec{m}$. In fact it is 
the core network $\mathcal B(\vec{m})$ for $\vec{m}$
whose construction is described in the proof of
Theorem~\ref{theo:arc-rich}. The network $N_0$
on the right of Figure~\ref{fig:simpli-traceback-cherry-picking-only} is the network $N(\vec{m})$
when initializing its construction with $\mathcal B(\vec{m})$. It
is obtained via the traceback
of $\sigma(\vec{m})$ by applying the cases indicated below the 
arrows. To be able to 
reuse the example to help illustrate Theorem~\ref{theo:arc-rich}, we represent one of the incoming arcs of each of the reticulation vertices of the networks that make up the figure as a thin, horizontal arc. Note that the leaf labels between the networks do not necessarily translate between the networks due to the applied renaming scheme for the elements of the indexing sets.

%{\bf Liam to update the following eample} To illustrate the construction of $N(\vec{m})$ from a ploidy profile $\vec{m}$, consider the ploidy profile $\vec{m}=(8,7,5,2,1)$ on $X=\{a,b,c,d,e\}$. Then 
%the sequence $(\vec{m}$, $(7,5,2,1,1)$, $(5,2,2,1,1)$, $(3,2,2,1,1)$, $(2,2,1,1,1)$, $(2,1,1,1)) $ is the simplification sequence $\sigma(\vec{m})$ and $\vec{m_t}=(2,1,1,1)$.  Let $B$ denote the
%phylogenetic network  obtained from the rooted caterpillar
%on $X_1=\{x_1,x_2,x_3,x_4\}$ with cherry $\{x_1,x_2\}$ and 
%$x_4$ a child of the root $\rho$ by adding a new root $\rho'$
%a subdivision vertex $s$ to the arc ending in $x_1$ and the arc $(\rho',s)$. Then the network $N(\vec{m})$ depicted in Figure~\ref{fig:simpli-traceback-cherry-picking-only} is the network
%returned by 
%Algorithm~\ref{alg:construction} when choosing $B$ as core network
%for $\vec{m}$ 

\begin{figure}[h]
	\centering
	\includegraphics[scale=0.5]{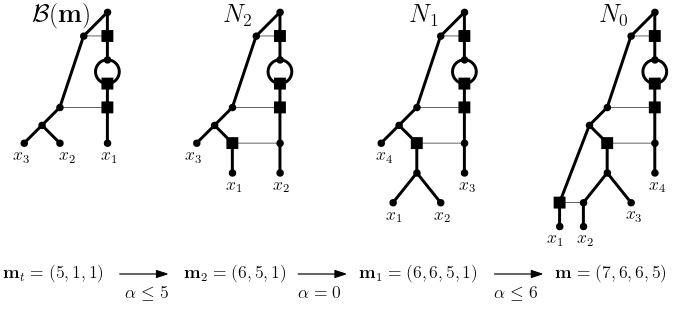}
	\caption{When reading from left to right, the construction of $N(\vec{m})$ obtained from the traceback through the simplification sequence $\sigma(\vec{m})$ for the ploidy profile $\vec{m}=(7,6,6,5)$ on $\{x_1,\ldots, x_4\}$. The ploidy profiles that make up
		$\sigma(\vec{m})$ are given at the bottom. The terminal element $\vec{m_t}$ of $\sigma(\vec{m})$ is the ploidy profile $(5,1,1)$ and the phylogenetic network $\mathcal B(\vec{m})$ on the left is a core network for $\vec{m}$.
		The cases that apply in each step of the traceback are indicated below the arrows between the four ploidy profiles that make up $\sigma(\vec{m})$. The thin horizontal arcs relate to the example illustrating Theorem~\ref{theo:arc-rich}.  
		\label{fig:simpli-traceback-cherry-picking-only}
	}
\end{figure}

We conclude this section by remarking that, by  \cite[Theorem 2]{HM21}, $N(\vec{m})$ employs the minimum number
of reticulation vertices to realize $\vec{m}$ provided (i) $\mathcal B(\vec{m})$ uses the minimum number
of reticulation vertices to realize the terminal element of the simplification sequence of $\vec{m}$, and (ii)
the case $\alpha>m_2'$ is never executed when constructing $N(\vec{m})$ from $\mathcal B(\vec{m})$ where 
$\alpha$ and $m_2'$ are as in the description of that case  
(see \orange{Figures~6 and 10 in \cite{HM21} for more on this} and the next section for an example).
%Section~\ref{sec:example} for an example). 

\section{Comparing realizations  of one and the same  ploidy profile}
\label{sec:comparison}
As indicated in  the previous section,
%Section~\ref{sec:weak-hgt-consistent}, 
a ploidy profile $\vec{m}$ might have more than one core network.
% and,vtherefore, might have more than one realization in terms
%of a phylogenetic network with horizontal arcs. 
This 
immediately begs the question of how different realizations of a ploidy profile might be. For phylogenetic trees and, more recently, for general rooted \orange{(beadless)} or unrooted phylogenetic networks this type of question has generally been addressed in the form of understanding their space. From a formal point of view, this space  which is called {\em tree space} in the case of phylogentic trees and {\em network space} in the case of \orange{rooted (beadless) or unrooted} phylogenetic networks is a graph.
Calling that graph $G$ then the vertices of $G$ are the phylogenetic trees or networks of interest and any two vertices of $G$ are joined by an edge if one can be transformed into the other using some graph-editing operation  such as the {\em Subtree Prune and Regraft operation (SPR)} for phylogenetic trees \cite{SS03} or one the operations described in \cite{EFM21,HMW16,J21,IJJM21}. 

None of the operations described in those papers however preserve, in general, our central requirement that a network is a realization of a ploidy profile. For technical reasons
	which will allow us to extend the idea of tree/network space to a space of ploidy profiles, we first need to extend
the notion of a phylogenetic network. To this end, we call
a phylogenetic network where different leaves are allowed to share the same label a {\em  multiple-labelled network}. 
In the form of, for example, multiple-labelled trees
such structures have already been used successfully in a polyploidization context \cite{OP11,R21}. For their usage in a more mathematical context see e.\,g.\,\cite{HS20} and the references therein. 
For example, \orange{consider the phylogenetic network depicted in Figure~\ref{fig:core-net-77-1-1-1}(ii). Then the graph obtained as follows is a multiple-labelled network. First, remove the
arc $(w,s_4')$ and one of the incoming arcs of $h_6$. Next, suppress
$h_6$ and its parent and add two further vertices  both of which we call $x_1$. Finally, 
add  an arc $(w,x_1)$  that ends in one of the two new vertices $x_1$  and an 
arc $ (s_4',x_1)$ that ends in the other} so that a cherry
on the multiset $\{x_1,x_1\}$  is generated. Since the number of directed
 paths from the root of the resulting graph to each of its leaves
 is not affected by this process, we extend the definitions of a reticulation vertex and  when a ploidy profile is realized by a
phylogenetic network to multiple-labelled networks in
the obvious way.

Armed with this, we are now ready to define ploidy profile space. Suppose $\vec{m}$ is a ploidy profile
\orange{on $X$}. Then we refer to the following graph as {\em ploidy profile space $\mathcal P(\vec{m})$} for $\vec{m}$. The vertex set of the graph is the set of all multiple-labelled networks that realize $\vec{m}$.
To be able to define the edge set of our graph, we require a further concept. 
We say that a multiple-labelled network $N'$ is obtained from
a multiple-labelled network $N$ via a
{\em split operation} if there exists a reticulation vertex $h$ of $N$ with parents $p_1$ and $p_2$ and child $c$ such that \orange{$(h,c)$ is a cut-arc} and $N'$ is obtained from $N$ as follows. First delete $h$ and its three incident arcs from $N$ and then make a copy of the subgraph \orange{of $N$ induced by the vertices of $N$ below $c$.}  Finally, add the arc $(p_2,c)$ as the  incoming arc to one of the two copies of $c$ and $(p_1,c)$ as the incoming arc of the other. 
We illustrate this operation in Figure~\ref{fig:split-operation}. 
\begin{figure}[h!]
	\centering
	\includegraphics[scale= 0.4]{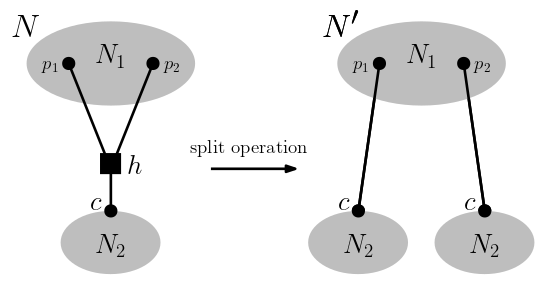}
	\caption{An illustration of the split operation applied to the reticulation vertex $h$ with parents $p_1$
		and $p_2$ and child $c$. $N_1$ and $N_2$ indicate
		parts of the multiple-labelled networks $N$ and $N'$ that are of no relevance to the discussion.
	\label{fig:split-operation}
	}
\end{figure}
Informally speaking, the split operation may be thought of as ``un-zipping'' a multiple-labelled network (see also \cite{PS15} for a related notion of ``unzipping'' a phylogenetic network). Choose  \orange{an edit distance
for comparing multiple-labelled trees that realize $\vec{m}$ such that the following graph is connected.
The  vertex set is the set of all multiple-labelled trees that realize $\vec{m}$ and any two multiple-labelled trees in that set are joined by an edge if their distance under the chosen edit distance is 1. For our next result (Proposition~\ref{prop:connected}), we are interested in edit distances for which this space is connected 
(see  \cite{HSSM11} for examples of such distances and also \cite{LE-MHM19} for some recent computational complexity results concerning them).}

Armed with \orange{the split} operation and choice of \orange{edit distance}, we continue our definition of ploidy profile space
for $\vec{m}$ as follows. We say that two distinct realizations 
	$N$ and $N'$ of $\vec{m}$ are joined by an edge if either $N'$ can be obtained from $N$ via a single split operation or $N$ and $N'$ are both
	multiple-labelled trees and 
	their distance under the chosen \orange{edit distance} is 1. 

\begin{proposition}\label{prop:connected}
	For any ploidy profile $\vec{m}$ on $X$ and any \orange{edit distance on the set of multiple-labelled trees realizing $\vec{m}$ such that the associated space of multiple-labelled trees is connected},
	ploidy profile space $\mathcal P(\vec{m})$ is connected.  
	\end{proposition}	

	\begin{proof} 
		\orange{Choose an edit distance $D$ such that the associated space of multiple-labelled trees  realizing $\vec{m}$ is connected.}
		Clearly, any realization $N$ of $\vec{m}$ can be transformed
		into a realization $N^+$ of $\vec{m}$ that does not contain	reticulation vertices that are above each other using a sequence of split operations. Since the vertex set of ploidy profile space is the set of multiple-labelled phylogenetic 
		networks \orange{that realize $\vec{m}$,}  it follows that $N^+$ can be transformed into
		a multiple-labelled tree \orange{that realizes $\vec{m}$} by using a further sequence of
		split operations. \orange{By assumption on $D$,}
		any multiple-labelled tree \orange{that realizes $\vec{m}$} can be transformed into another multiple-labelled tree \orange{which also realizes $\vec{m}$} via a sequence $\chi$ of multiple-labelled trees \orange{that all realize $\vec{m}$,} such that any two
		consecutive multiple-labelled trees in $\chi$ have distance \orange{1 under $D$. Hence, }
		$\mathcal P(\vec{m})$ is connected.
		\hfill{$\Box$}

		\end{proof}
		
As an immediate consequence of Proposition~\ref{prop:connected}, we obtain a distance measure for realizations of a ploidy profile \orange{$\vec{m}$}. More precisely, choose an \orange{edit distance} $D$ for comparing two multiple-labelled trees \orange{that realize $\vec{m}$ such that the space associated to $D$ is connected}. Then we define the distance $D_{\mathcal P(\vec{m})}(N,N')$ of any two realizations $N$ and $N'$ of $\vec{m}$ to be the length of a shortest path in $\mathcal P(\vec{m})$ that joins $N$ and $N'$. We refer the interested reader to Section~\ref{sec:example} for 
 an example where we compute an upper bound on this distance for a real biological dataset.

\section{A core network with horizontal arcs}
%\section{Weak HGT-consistent labellings and core networks}
\label{sec:weak-hgt-consistent}
Although undoubtedly useful, phylogenetic networks on their own do not provide information as to whether or not two species suspected of hybridization have existed at the same point in time.
To add this type of realism to phylogenetic networks, so called 
time stamp maps can be used. Subject to some constraints
such maps assign a non-negative real number to every vertex of a phylogenetic network (see e.\,g.\,\cite{BS06,FS15,IJJMZ20}). 
As is well-known, not
every phylogenetic network admits a time stamp map. However those that do enjoy the attractive property that arcs whose both end vertices have been assigned the same time stamp can be drawn horizontally to indicate that the ancestral species represented by their end vertices have existed at the same time. 

To be able to extend the notion of a time stamp map to ploidy profiles, we start with the definition of a certain time stamp map for (beadless) 
phylogenetic networks that originally appeared in \cite{IJJMZ20}.
Suppose that  $N$ is a beadless phylogenetic network on a set $X$ with at least two elements. Then
 a map $t:V(N)\to\mathbb R_{\geq 0}$ 
 %from the set of vertices of $N$ to the set of positive real numbers 
 is called a {\em HGT-consistent labelling} of $N$  
 if the following properties hold:
\begin{enumerate}
	\item[(P1)] For all arcs $(u,v)$ of $N$, we have that $t(u)\leq t(v)$
	if $v$ is a reticulation vertex and, otherwise, that $t(u)<t(v)$.
	\item[(P2)] For each vertex $u$ that is not a leaf of $N$ there exists
	a child $v$ such that $t(u)<t(v)$.
	\item[(P3)] For each reticulation vertex $v$ 
	%that is not a vertex in a bead 
	of $N$ there exists precisely one
	parent  $u$
	such that $t(v)=t(u)$.
	%\item[(P4)] For each reticulation vertex $v$ that forms a bead
	%with a vertex $u$ we have that $t(u)=t(v)$.
\end{enumerate}
Informally speaking, Property~(P1) means that time is moving forward, from 
the root of the network to its leaves. Property~(P2) implies that every
ancestral species $v$ has given rise to at least one species that did not exist at the same time as 
$v$. Finally,
Property~(P3) implies for every reticulation vertex $v$ that the ancestral species represented by $u$
must have existed at the same time as the species represented 
by $v$.  Examples of (beadless) phylogenetic networks that admit a HGT--consistent labelling include 
{\em stackfree} phylogenetic networks, that is, (beadless) phylogenetic networks that have no arcs for which both end vertices are reticulation vertices (see \cite{BESS21} for more on such networks). 
It should however be noted that there exist
(beadless) phylogenetic networks that admit such a labelling which might not be stackfree.

Since the definition of a HGT-consistent labelling 
of a beadless phylogenetic network $N$ does not rely on the assumption that $N$ has no
 beads, we extend it to our type of phylogenetic network by dropping the ``beadless'' requirement
 and qualifying Property~(P3) by excluding reticulation vertices that are contained in beads.
 In combination, Property~(P2) and the thus adjusted Property~(P3) imply that there cannot have existed
an ancestral species $v$ such that $v$ is involved
in a polyploidization event and, at a later point in time, 
one of its parents,  $p$ say, hybridizes with the unique child $u$ of $v$. Put differently, we cannot simultaneously have all three arcs $(v,u)$, $(p,v)$, and $(p,u)$ and  $v$ is a reticulation vertex. 
%To help illustrate the effect this drop has,
%we include the presentation of a HGT-consistent labelling
%in our discussion of the example which we use to illustrate our main result of this section
%(Theorem~\ref{theo:arc-rich}), i.\,e.\,the phylogenetic network depicted in Figure~\ref{fig:core-net-77-1-1-1}.
%
%Calling a directed graph with four distinct vertices $a$, $b$, $c$, and $d$ and  arcs $(b,a)$, $(c,b)$, $(c,d)$, $(d,a)$ and $(b,d)$ a {\em banded bead} then our extension of the concept
%of a HGT-consistent labelling to a weak HGT-consistent labelling
%is motivated by the observation that no phylogenetic network that contains a banded bead as a subgraph can admit a HGT-consistent labelling. 
%

In view of the aforementioned combined effect of Properties~(P2) and (P3), we also say that a phylogenetic network $N$ admits a {\em weak HGT-consistent labelling} $t$ 
if $t$ is a map from the vertex set of $N$ to the set of non-negative real numbers 
such that Properties~(P1) and (P2) hold and Property~(P3)
is weakened to 
\begin{enumerate}
\item[(P3')] there exists at most one reticulation vertex $v$
with distinct parents $u_1$ and $u_2$ 
with $u_2$ below $u_1$ 
such that $v$ does not satisfy Property~(P3) i.\,e.\, $t(u_i)\not=t(v)$, for all
$i=1,2$.
\end{enumerate}
%
%As in the case of a HGT-consistent labelling for a phylogenetic network, we include an example of a weak HGT-consistent
%labelling of a phylogenetic network in the discussion of our example to illustrate Theorem~\ref{theo:arc-rich}. 
As a first observation, note that a HGT-consistent labelling of a phylogenetic network is also a weak HGT-consistent labelling for that network.

To be able to state Theorem~\ref{theo:arc-rich}
which is concerned with clarifying the structure of ploidy profiles that admit a HGT-consistent
labelling or a weak HGT-consistent labelling, we require the concept of a binary representation
of a positive integer $m$. This representation essentially records for the representation of $m$ as a sum $\sum_{i=0}^l 2^q$ of ``powers of two'' the vector of exponents.
More formally, we define the {\em binary representation} of 
a positive integer $m=\sum_{j=1}^k 2^{i_j}$ to be the vector $(i_1,\ldots, i_k)$, $k\geq 1$, with $i_{j-1}>i_j$, for
all $2\leq j\leq k$, and $i_k\geq 0$. Note that although related,
the binary representation of $m$ is not the bit-wise representation of $m$. For example, for $m=77=2^6+2^3+2^2+2^0$ the binary representation is $(6,3,2,0)$ whereas the
bit-wise representation is $(1,0,0,1,1,0,1)$. 

We say that a strictly simple ploidy 
profile $\vec{m}=(m_1)$ is {\em arc-rich} if the dimension of the binary representation of $m_1$ is at least two. Furthermore, we call a ploidy profile $\vec{m}$ {\em practical} 
if either $\vec{m}$ is simple but not strictly simple or $\vec{m}=(m_1)$ and $m_1$ is of the form 
$m_1=2^l$, some $l\geq 1$.  For example, the
ploidy profile $\vec{m}=(77)$ is arc-rich but not practical
since the binary representation of $77$ is the vector $(6,3,2,0)$.   

\begin{theorem} \label{theo:arc-rich}
	Suppose $\vec{m}$ is a ploidy profile. If the terminal element $\vec{m_t}$ of the simplification sequence for $\vec{m}$ is practical then
	there exists a core network for $\vec{m}$ that admits a HGT-consistent labelling. Otherwise, $\vec{m_t}$ is arc-rich and there exists a core network for $\vec{m}$ that admits a weak HGT-consistent labelling.
\end{theorem}
\begin{proof}
For ease of readability, we split 
 the proof into three sections, as indicated below. We start with  introducing a further concept. Suppose $T$ is a
 phylogenetic tree on $X=\{x_1,\ldots, x_n\}$, some $n\geq 2$.
 Then we call $T$ a {\em caterpillar tree (on $X$)} if the elements of $X$ can be relabelled in such a way that $T$ has a single cherry and that cherry is $\{x_{n-1},x_n\}$. If $n\geq 3$ then $x_1$ is a leaf
 that is a child of the root $\rho$ of $T$, and every vertex
 on the path from $\rho$ to the shared parent \orange{$f$} of $x_n$ and $x_{n-1}$ other than $\rho$ and \orange{$f$} has a child that is a leaf.
 For ease of presentation, we assume that the other child of the parent \orange{$f'$} of \orange{$f$} is $x_{n-2}$, the other child of the parent 
  of \orange{$f'$} is $x_{n-3}$
 and so on. 
 %For example, in Figure~\ref{fig:core-net-77-1-1-1} the phylogenetic tree on  $X=\{x_2,x_3,x_4\}$ used in the construction of the core network $\cal B(\vec{m})$  for the ploidy profile $\vec{m}=(77,1,1,1)$ is a caterpillar tree on $X$.
 
 For the remainder of the proof, assume that $\vec{m}$ is a 
 simple ploidy profile  (see Figure~\ref{fig:core-net-77-1-1-1}   for an illustration of our
 constructions for the ploidy profile $\vec{m}=(77,1,1,1)$ on $X=\{x_1,\ldots, x_4\}$). 
 Since a core network for $\vec{m}$ realizes the terminal element
 $\vec{m_t}=(m_1,\ldots, m_n)$ of the simplification sequence for $\vec{m}$ 
 and $\vec{m_t}$ is simple, we need to consider the cases that
 $\vec{m_t}$ is strictly simple and that $\vec{m_t}$ is not strictly simple. Let $X=\{x_1,x_2,\ldots, x_n\}$ denote the set that indexes $\vec{m_t}$. \orange{Recall that, by convention,} $x_i$ indexes
 $m_i$, for all $1\leq i\leq n$.  \\

% \noindent {\em Proof of Theorem~\ref{theo:arc-rich}:}
 
  \noindent {\em Construction of the core network $\mathcal B(\vec{m})$:}
Assume first that $\vec{m_t}$ is a 
 strictly simple ploidy profile. Then $\vec{m_t}=(m_1)$  and $X=\{x_1\}$.
 Let $\vec{i}=(i_1,\ldots, i_k)$, some
 $k\geq 1$, denote the binary representation of $m_1$. 
 Note that  $i_1\geq 1$ because $m_1\geq 2$. Then we first
 construct a beaded tree $B(i_1)$ that realizes the 
 strictly simple ploidy profile $(i_1)$ by taking $i_1$ beads
 $B_1, B_2,\ldots B_{i_1}$ and, provided $i_1\geq 2$, adding for all $1\leq i\leq i_1-1$
 an arc from the reticulation vertex $h_i$ of $B_i$ to the tree vertex
 of $B_{i+1}$. To the resulting  graph we then add the 
 vertex $x_1$ and an arc $(h_{i_1},x_1)$ to obtain $B(i_1)$.
 If $m_1=2^l$, some positive integer $l$, then we define $\mathcal B(\vec{m})$ to be $B(i_1)$.

So assume that there exists no positive integer $l$ such that
$m_1=2^l$. Then $k\geq 2$. Let $B(i_1,i_k)$ denote the phylogenetic network obtained from $B(i_1)$ by subdividing one of the two outgoing arcs of the root $\rho$ of $B(i_1)$
 by a subdivision vertex $s_k$, the outgoing arc of the 
 reticulation vertex in $B(i_1)$ that has precisely $i_k$ reticulation vertices strictly below it 
 by a vertex $s_k'$ and adding the arc $a_k=(s_k,s_k')$. 
 If $k=2$, then $\mathcal B(\vec{m})$ is $B(i_1,i_2)$. 
 
 Finally, assume that $k\geq 3$. 
Then we first construct the graph $B(i_1,i_k)$. Next, 
 we subdivide the arc $a_k$ by $k-2$ vertices $s_2,\ldots, s_{k-1}$ such that $(s_k,s_2)$ is an arc and $s_j$ is the parent
 of $s_{j+1}$ for all $2\leq j \leq k-2$. For all
 $2\leq j\leq k-1$, we next
 subdivide the outgoing arc of the reticulation vertex of $B(i_1,i_k)$
 that has precisely $i_j$ reticulation vertices of $B(i_1)$ strictly below 
 it by a vertex $s_j'$. Finally, we add for all such $j$
 the arc $a_j=(s_j, s_j')$ and  denote the resulting graph by 
 $\mathcal B(\vec{m})$ in this case.
 By construction, $\mathcal B(\vec{m})$ is a phylogenetic network on $x_1$ 
 that realizes $\vec{m_t}$ in either of these cases for $k$.
 
So assume that $\vec{m_t}$ is not strictly simple. Then 
$m_j=1$, for all $2\leq j\leq n$. Using the same notation
as before, we first construct the network $\mathcal B(\vec{m'})$ for the ploidy profile $\vec{m'}=(m_1)$. If $k=1$
then there exists some positive integer $l$ such that $m_1=2^l$. 
Hence, $\mathcal B(\vec{m})$ is $B(i_1)$ and  
we subdivide one of the outgoing arcs of the root of
$B(i_1)$ by a vertex $w$. So assume that $k\geq 2$. If $k=2$  then 
$\mathcal B(\vec{m'})$ is $B(i_1,i_2)$ and  we
subdivide the arc $a_2$ of $\mathcal B(\vec{m'})$ by a vertex
$w$. So assume $k\geq 3$. Then  we
subdivide the arc $a=(s_{k-1},s_k')$ of $\mathcal B(\vec{m'})$
by a vertex $w$.  In either of these cases for $k$ 
we then attach the caterpillar tree $T$ on
$\{x_2,\ldots, x_n\}$ to $\mathcal B(\vec{m'})$ via an arc from $w$ to the 
root of $T$ in case $n\geq 3$. If $n=2$ then we attach the vertex $x_n$ via the pendant arc $(w,x_n)$.
 By construction, the resulting graph is a
 phylogenetic network that realizes $\vec{m_t}$,
 and it is the network $\mathcal B(\vec{m})$ in this final case
 for $\vec{m_t}$.\\

 \noindent {\em Construction of a HGT-consistent labelling for $\mathcal B(\vec{m})$ in case $\vec{m_t}$ is practical:}
% \label{sec:assertion-i}
Assume first that $\vec{m_t}$ is strictly simple. Then $m_1=2^{i_1}$
and $\mathcal B(\vec{m})$ is $B(i_1)$. Hence,
there exists a directed path $P:v_0=\rho, v_1, v_2,\ldots, v_q=x_1$ 
from the root $\rho$ of $\mathcal B(\vec{m})$ to $x_1$ once one arc
has been removed from each bead of $\mathcal B(\vec{m})$. Note that $P$ contains
vertices with indegree and outdegree one and that $V(P)$
is the vertex set of $\mathcal B(\vec{m})$. 
Consider the map $t:V(P)\to \mathbb R_{\geq 0}$
defined by putting $t(v_0)=0$ and $t(v_{j+1})= t(v_j)+1$,
for all other $0\leq j\leq q-1$.  By construction,
it follows that  $t$ is a HGT-consistent labelling for $\mathcal B(\vec{m})$ in this case.

So assume that $\vec{m_t}$ is not strictly simple.
Then $\vec{m_t}$ must be simple because it is the 
terminal element of the simplification sequence for $\vec{m}$.
Let $P:v_0=\rho, v_1, v_2,\ldots, v_q=x_1$ denote the directed path in $\mathcal B(\vec{m})$ from $\rho$ to $x_1$ obtained by removing (i) 
the caterpillar tree on $\{x_2,\ldots, x_n\}$ and the incoming arc of its root in case $n\geq 3$ and $x_n$ and its pendant arc if $n=2$, (ii) for all $2\leq j\leq k$, the vertices $s_j$ and their incident arcs, and  (iii) one of the two arcs in every bead. Let $t_P:V(P)\to \mathbb R_{\geq 0}$ be defined as the map 
$t$ in the previous case. 

Consider the map $t:V(\mathcal B(\vec{m}))\to \mathbb R_{\geq 0}$ defined by putting
$t(v)=t_P(v)$ for all vertices $v$ of $\mathcal B(\vec{m})$ that are also vertices
on $P$.  So let $v$ be a
vertex in $\mathcal B(\vec{m})$ that does not lie on $P$. If $v=s_k$
then put $t(v)=t(h_1)$ and if $v=w$ then
put $t(v)=t(s'_k)$. For all $2\leq j\leq k-1$, put
$t(s_j)= t(s'_j)$. Note that this does not violate Properties~(P1)-(P3) since, for all $2\leq j \leq k-1$, we have 
$t(s_j')<t(s_{j+1}')$ and $t(s_k)<t(h_1)<t(s_2')=t(s_2)$. 
If $n\geq 3$ then, for all 
	$1\leq j\leq n-1$, let $w_j$ denote that parent of
	the leaf $x_{j+1}$ of  $T$. Put $t(w_1)=t(w)+1$
	and, for all $1\leq j\leq n-2$, put $t(w_{j+1})=t(w_j)+1$.
	Finally, choose a value $\chi>t(w_{n-1})$ and put 
	$t(x_j)=\chi$, for all $2\leq j\leq n$.
	By construction, $t$ respects Properties~(P1)-(P3), and so 
	$t$ is a HGT-consistent labelling for $\mathcal B(\vec{m})$.
	If $n=2$ then we proceed in a similar manner in that we put $t(x_n)=t(w)+1$.\\

 \noindent {\em Construction of a weak HGT-consistent labelling for $\mathcal B(\vec{m})$ in case $\vec{m_t}$ is not practical:}
 If $\vec{m_t}$ is not practical it must be arc-rich as
$\vec{m_t}$ is the terminal element of the simplification 
sequence of $\vec{m}$.
It suffices to note that a weak HGT-consistent
labelling can be constructed as in the case of a HGT-consistent labelling
%Section~\ref{sec:assertion-i} 
noting that the only reticulation vertex of $\mathcal B(\vec{m})$ that violates Property~(P3)
is $s_k'$.
Thus, $t$ satisfies Property~(P3') and so $\mathcal B(\vec{m})$ admits a weak HGT-consistent labelling. 
\hfill{$\Box$}
\end{proof}

As mentioned in the proof of Theorem~\ref{theo:arc-rich}, we \orange{next} illustrate the construction of
$\mathcal B(\vec{m})$ for the ploidy profile $\vec{m}=(77,1,1,1)$ on $X=\{x_1,x_2,x_3,x_4\}$.
Then the vector $\vec{i}=(6,3,2,0)$ is a binary representation for $77$ and 
the phylogenetic network  depicted in
Figure~\ref{fig:core-net-77-1-1-1}(i) is $\mathcal B(\vec{m'})$ where $\vec{m'}=(77)$. Clearly, the
phylogenetic network $\mathcal B(\vec{m})$ depicted in Figure~\ref{fig:core-net-77-1-1-1}(ii) 
obtained from $\mathcal B(\vec{m'})$ by adding the leaves $x_2$, $x_3$, and $x_4$ as indicated realizes
$\vec{m}$ and admits a HGT-consistent labelling. Since the actual time stamp values are of no interest to our discussion, we indicate arcs for which both end vertices have the same time stamp under a HGT-consistent labelling in terms of 
horizontal arcs. 
 \begin{figure}[h]
	\centering
	\includegraphics[scale=0.3]{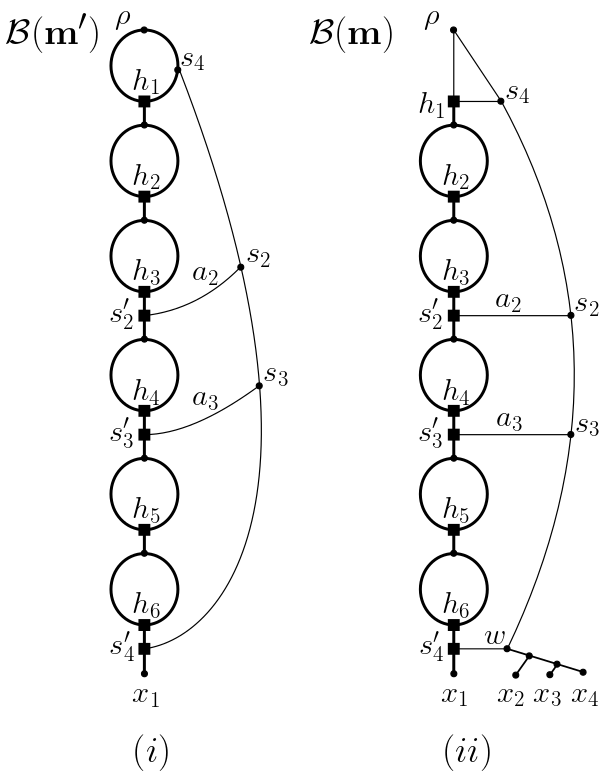}
	\caption{(i) The core network ${\mathcal B}(\vec{m'})$
		for the strictly simple ploidy profile $\vec{m'}=(77)$ on $X=\{x_1\}$. (ii) The core network $\mathcal B(\vec{m})$ for the simple ploidy profile 
		$\vec{m}=(77,1,1,1)$ on $\{x_1,x_2,x_3,x_4\}$ obtained from $\mathcal B(\vec{m'})$. Alternative core networks for $\vec{m}$ can be obtained from $\mathcal B(\vec{m'})$ by subdividing non-bold arcs and attaching
		the remaining elements of $X$ as phylogenetic trees
		on subsets of $X$ or individually (ensuring that the arc
		$(s_3,s_4')$ is subdivided at a least once as otherwise the resulting phylogenetic network does not admit a HGT-consistent labelling because Property~(P3) is violated).
		\label{fig:core-net-77-1-1-1}
	}
\end{figure}
%

%Intriguingly, the construction of $\mathcal B(\vec{m})$ implies that it also tree-child in case $\vec{m}$ is simple or strictly simple and not arc-rich where we extending the definition of 
%of a tree-child (beadless) phylogenetic network to our types of phylogenetic networks in the canonical way. More precisely,
%we say that a phylogenetic network
%is {\em tree child} if for every one of its vertices there exists a directed path $P$ to a leaf so that after removing from every bead crossed by $P$ one arc and suppressing the resulting indegree and outdegree one vertices  no vertex on $P$
%other than potentially $v$ is a reticulation vertex in a bead
%is a reticulation vertex. In consequence $\mathcal B(\vec{m})$
%is not tree-child if $\vec{m}$ is strictly simple and arc rich.
%This implies that, in general, the network generated from the
%core network of a ploidy profile in not tree-child.

As indicated in Figure~\ref{fig:twocores12},  alternative 
choices of a core network for a ploidy profile $\vec{m}$ are conceivable in the sense that it might not be obtained 
by starting with a binary representation of the first component
of $\vec{m}$. Furthermore and perhaps not surprisingly, 
a core network $N$ for $\vec{m}$ generally admits more than one HGT-consistent labelling
in the sense that an alternative HGT-consistent labelling for $N$ might assign for a reticulation vertex $v$
the same time stamp as for $v$ to a different parent of $v$.

\begin{figure}[h]
	\centering
	\includegraphics[scale=0.35]{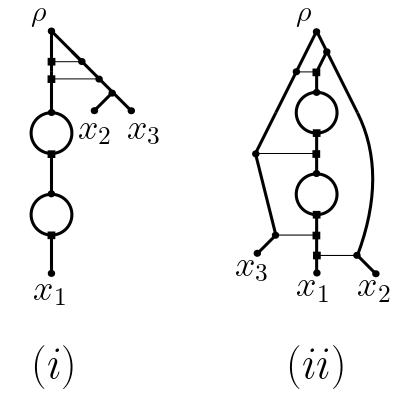}
	\caption{(i) The realization $\mathcal B(\vec{m})$ of the practical ploidy profile 
		$\vec{m}=(12,1,1)$ on $X=\{x_1,x_2, x_3\}$. (ii) A core network for $\vec{m}$
		that is not of the form $\mathcal B(\vec{m})$. 
		\label{fig:twocores12}
	}
\end{figure}  

The fact that the simplification sequence of a ploidy profile $\vec{m}$ is obtained by taking differences 
of the first two consecutive components of 
$\vec{m}$ implies that, in general, properties of ploidy profiles
do not get inherited by ploidy
profiles obtained from $\vec{m}$. For certain types of ploidy profiles this is however not the case as the following consequence of Theorem~\ref{theo:arc-rich} shows.

\begin{corollary}
	\label{cor:special-case}
	Suppose that $\vec{m}$ is a ploidy profile. Then the following holds.
	\begin{enumerate} 
	\item[(i)] If $\vec{m}$ is of the form $(n,n-1,n-2,\ldots, 1)$, $n\geq 3$,  then, for any ploidy profile obtained from $\vec{m}$ by removing at most one of its components, there exists a core network that admits a HGT-consistent labelling.
	\item[(ii)] If $\vec{m}=(m_1,\ldots, m_n)$ then there exists a core network for $(m_1,\ldots, m_n,1)$ that admits a HGT-consistent labelling.
	\item[(iii)] If $\vec{m}=(m_1,\ldots, m_n)$ has a core network that admits a HGT-consistent labelling then the ploidy profile $(2^im_1, \ldots, 2^im_n,2^{i-1},2^{i-2}\ldots, 2^0)$, $i\geq 1$, has a realization that also admits such a labelling. 
	\end{enumerate}
\end{corollary}
\begin{proof}
(i) Let $\vec{m'}$ denote a ploidy profile obtained from $\vec{m}$
 	as described in the statement of the corollary. Then 
	\orange{the difference between any two consecutive component values in $\vec{m}$ is $1$ if no component is removed from $\vec{m}$ (i.\,e.\,$\vec{m}=\vec{m'}$) or if a component is removed from $\vec{m}$ to obtain $\vec{m'}$ whose value is not 2. In either of these two cases, it follows that the terminal element $\vec{m'_t}$  of the 
		simplification sequence for $\vec{m'}$ is of the form $(2,1,\ldots, 1)$. If the component with value 2 is removed to obtain $\vec{m'}$ from $\vec{m}$ then the
	 terminal element $\vec{m'_t}$ of $\sigma(\vec{m'})$ is of the form  $(3,1,\ldots, 1)$ as that ploidy profile is simple. In either of these cases,  $\vec{m'_t}$ is practical.}
	Applying
 	Theorem~\ref{theo:arc-rich} implies the result. 
 	
	(ii) To see the assertion, it suffices to note that the terminal element
 	of the simplification sequence for $\vec{m'}=(m_1,\ldots, m_n,1)$ is practical
 	\orange{because it is of the form $(m,1, \ldots, 1)$, some $m\geq 2$.}

(iii) Put $\vec{m'}=(2^im_1, \ldots, 2^im_n,2^{i-1},2^{i-2}\ldots, 2^0)$. Let
$\mathcal B(\vec{m})$ initialize the construction of
 $N=N(\vec{m})$. Since, by assumption, $\mathcal B(\vec{m})$ admits
 a HGT-consistent labelling, it follows by construction that
 $N$ also admits such a labelling.
 Let $t:V(N)\to\mathbb R_{\geq 0 }$ denote a HGT-consistent labelling for $N$. 
 
 Next, consider the ploidy profile $\vec{m''}=(2^{i-1},2^{i-2}\ldots, 2^0)$ on $X$ where $x_{n+j}$ indexes
 $2^{i-j}$,	for all $1\leq j\leq i$. Then construct the core network $\mathcal B(\vec{m''})$
 for $\vec{m''}$. Since $\vec{m''_t}=(2,1)$ and therefore is not strictly simple 
 $\mathcal B(\vec{m''})$ admits a HGT-consistent
 labelling. Initializing the construction of 
 $N''= N(\vec{m''})$ with $\mathcal B(\vec{m''})$ implies that
 $N''$  also admits a HGT-consistent labelling $t'':V(N'')\to\mathbb R_{\geq 0 }$.
 
 Next, construct a realization $N'$ for $\vec{m'}$ from $N$ and $N''$ by subdividing
 	the incoming arc of $x_{n+1}$ by \orange{two new vertices $s$ and $s'$
 	such that $s'$ is below $s$. Next, add a further vertex $s''$ and the arcs 
 	$(s,s'')$, $(s',s'')$, and $(s'',q)$ where $q$ is}  the root of $N$ to obtain $N'$.
 	To obtain a HGT-consistent labelling $t':V(N')\to \mathbb R_{\geq 0}$
 	for $N'$ put $t'(v)=t''(v)$ for all vertices $v$
 	of $N'$ that are also vertices in $N''$. Next, choose 
 	a value $t''(p)<\omega< t''(x_{n+1})$ where $p$ is the parent
 	of $x_{n+1}$ in $N''$ and put \orange{$t'(s)=\omega$,
 		$t'(s')=t'(s'')=\omega+\epsilon$, some $\epsilon >0$ sufficiently small, and} 
 		  $t'(q)=t'(s)+1$. Finally, for all vertices
 	$v$ in $N$ put $t'(v)=t(v)+ t'(q)+1$. Since $t''$ 
 	is a HGT-consistent labelling for $N''$
 	and $t$ is such a labelling for $N$ it follows by construction
 	that $t'$ is a HGT-consistent labelling for $N'$.
	\hfill{$\Box$}

\end{proof}

To help illustrate Corollary~\ref{cor:special-case}(iii), consider the
ploidy profile $\vec{m'}=(40,24,8,$ $4,2,1)=(2^3\times 5, 2^3\times 3, 2^3\times 1,2^2,2^1,2^0)$ on $X=\{x_1,\ldots, x_6\}$. Then $i=3$
and $\vec{m}$ is the ploidy profile $(5,3,1)$ on $\{x_1,x_2,x_3\}$. 
Hence, $\vec{m_t}=(2,1,1)$.  By Theorem~\ref{theo:arc-rich},
$\mathcal B(\vec{m})$ admits a HGT-consistent labelling
because $\vec{m_t}$ is practical.
Initializing the construction of $N=N(\vec{m})$ with $\mathcal B(\vec{m})$ implies that $N$ also admits a HGT-consistent
labelling. The part of the
network $N'$ pictured in Figure~\ref{fig:powersoftwo} 
\begin{figure}[h]
	\centering
		\includegraphics[scale=0.3]{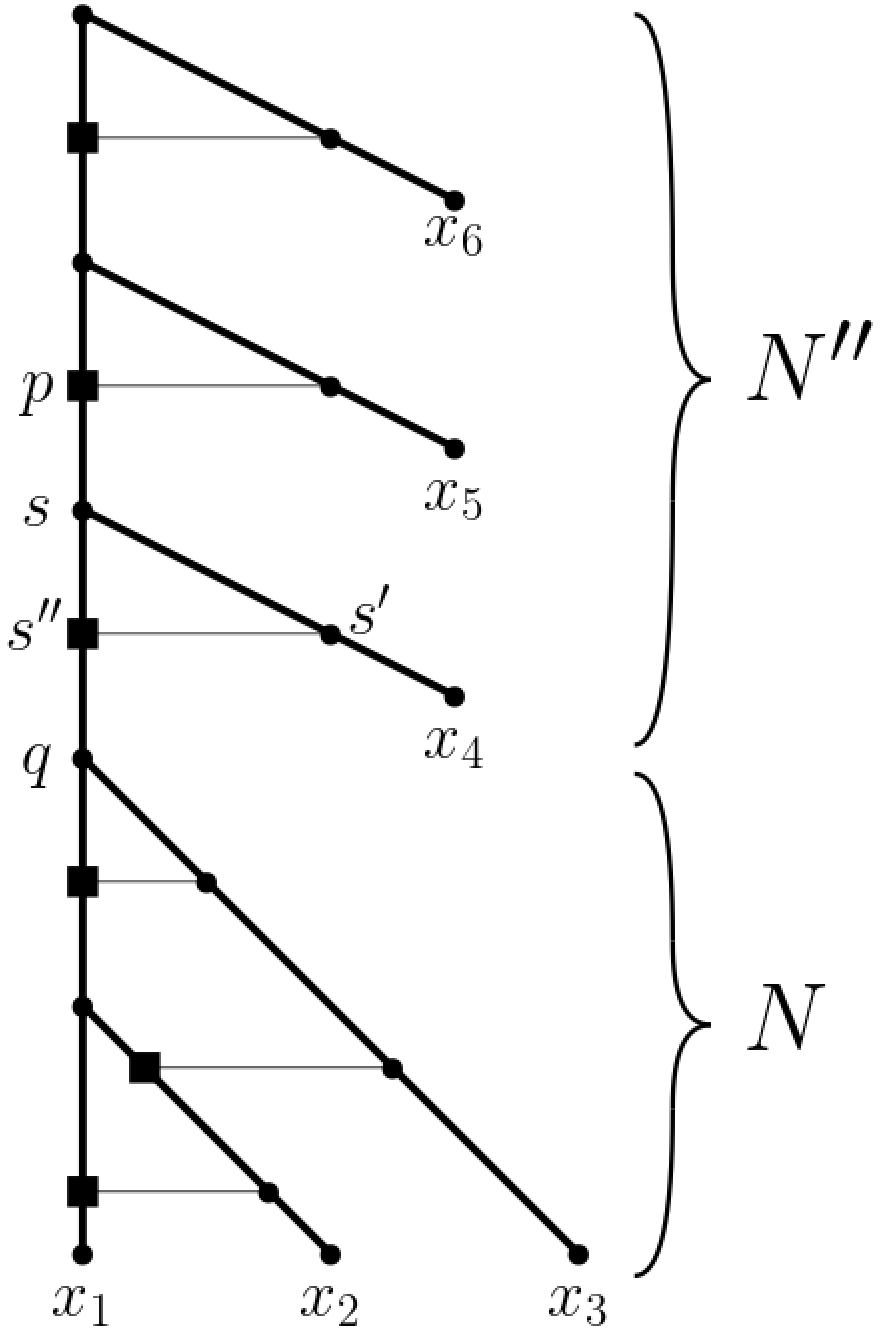}
	\caption{The realization $N'$ of the ploidy profile $\vec{m'}=(40,24,8,4,2,1)$ in terms of a phylogenetic network
		with horizontal arcs.  
		\label{fig:powersoftwo}
	}
\end{figure}  
that is labelled $N$
is that realization with a HGT-consistent labelling indicated 
in terms of horizontal arcs.  
The part of $N'$ labelled $N''=N(\vec{m''})$
is a realization of the ploidy profile $\vec{m''}=(4,2,1)$
on $\{x_4,x_5,x_6\}$ once \orange{the} three incident arcs \orange{of $s''$} are ignored and \orange{$s''$ and} the resulting vertices with indegree and outdegree one are suppressed. By construction, $N'$ is a realization 
of $\vec{m'}$ that admits a HGT-consistent labelling (again indicated in terms of horizontal arcs). 

We conclude this section with remarking that Corollary~\ref{cor:special-case}(ii) is of particular interest from a ``ghost species''  point of view in that the element $x_{n+1}$
indexing the last component of $(m_1,\ldots, m_n,1)$ could represent a taxon with ploidy level one   that has not been sampled yet  \orange{(see e.\,g.\,\cite{SBPVdHCPRR22} for the case of banana).}

\section{When is $N(\vec{m})$ a tree with additional horizontal arcs?}
%\section{Orchards, base trees, and ploidy space}
%\section{Weakly orchard ploidy profiles and beyond}
\label{sec:weakly-orchard-and-beyond}
As was established in \cite[Section 2.1]{IJJM21}, beadless
phylogenetic networks that admit a HGT-consistent labelling  are precisely the ones that admit a so called complete cherry reduction sequence. These types of sequences essentially record how to transform
a \orange{(beadless)}  phylogenetic network into a single vertex by applying only
operations on pairs of leaves, provided this is possible. In view of \cite[Theorem 1]{IJJM21} and \cite{ESS19}, their attraction lies in the fact that they can be used to quickly check if a given \orange{(beadless)} phylogenetic network can be represented with horizontal arcs without having to find a HGT-consistent labelling for it first. Therefore it is of interest to see if an analogous result holds for our types of phylogenetic networks. To be able to shed light into this question, we
first need to  extend the concept of a cherry reduction sequence to our types of phylogenetic networks. For this we require further terminology. 

Suppose that $N$ is a phylogenetic network on $X$. Assume first that $X$ has at least two elements and that $a$ and $b$ are distinct elements in $X$. 
If $\{a,b\}$ is a cherry of $N$ then 
we refer to the operation of deleting $b$ and its incoming arc and suppressing 
the resulting vertex of indegree and outdegree one
as \textit{reducing} the cherry $\{a,b\}$ by $b$. 
We denote this operation by $reduce(a,b)$.
Note that if the joint parent of 
$a$ and of $b$ is the root $\rho$ of $N$
and $N$ therefore has leaf set $\{a,b\}$, then this
operation also includes post-processing the resulting graph
by collapsing the unique arc from $\rho$ to $a$ 
to obtain the single vertex $a$. 
%For purely technical reasons, we extend for the remainder of this section our definition of a phylogenetic network to an isolated vertex $v$ by also calling $v$  a phylogenetic network (on $\{v\}$).  
If  $a$ and $b$ form a reticulated cherry of $N$ such that $b$ 
is the reticulation leaf then we refer to the operation of deleting the reticulation arc 
and suppressing the resulting vertices of indegree and 
outdegree one as \textit{cutting} the cherry  $\{a,b\}$. 
We denote this operation by $cut(a,b)$. 
For example, deleting the thin arc incident with the 
parent of $x_1$ in the network $N_0$ pictured in
Figure~\ref{fig:simpli-traceback-cherry-picking-only} is the cutting operation
$cut(x_2,x_1)$. Deleting the leaf $x_1$ in the network
$N_1$ pictured in that figure is the reducing operation
$reduce(x_2,x_1)$. 

Finally assume that $a$ is the sole element of $X$. Then we refer to the set
$\{a\}$ as a {\em type-1 degenerate cherry} if the parent of $a$
is the reticulation vertex in a bead $B$. In this case we call
the operation of removing one of the two arcs of $B$,
suppressing resulting vertices with indegree and outdegree one,
and also collapsing the unique outgoing arc of the tree
vertex of $B$ if that has rendered it a vertex of outdegree one
as {\em simplification} of $a$. We denote this operation
by $simp(a)$. Furthermore, we refer to the set $\{a\}$ as a {\em type-2 degenerate cherry} if $a$ has a parent $p$ that
is a reticulation vertex and either \orange{(i)} precisely one of the parents \orange{$q_1$ and $q_2$} of $p$
is the reticulation vertex of a bead or \orange{(ii) there exists a further vertex $q$ 
	such that  $N$ also contains (a) the three arcs $(q,q_1)$, $(q,q_2)$, and $(q_1,q_2)$, or (b) the arc $(q_1,q)$ and  a pair of arcs from $q$ to $q_2$. Assuming that Case (ii)  holds then,}
we refer to the operation of deleting the arc \orange{$(q_1,p)$ (Case (a)) and deleting one of the arcs from $q$ to $q_2$ (Case (b))} and in each case suppressing the two resulting vertices of indegree and outdegree 1 as {\em trimming} of $a$.

We denote this operation as
$trim(a)$. For example for the 
network pictured in Figure~\ref{fig:core-net-77-1-1-1}(i), the trimming operation $trim(x_1)$
consists of deleting the arc $(s_3,s_4')$ and suppressing the vertices  $s_4'$ and $s_3$.  Removing one of the two arcs in the bead in the network
depicted in Figure~\ref{fig:twocores12}(i) that contains the 
parent of $x_1$
is the simplification operation $simp(x_1)$.

It is easily seen that the operations of reducing a cherry and cutting a reticulated cherry both result in a phylogenetic 
network where, for technical reasons, we refer in this context to an isolated vertex $a$ also as a phylogenetic network on $\{a\}$. Collectively, these two operations are usually referred to as  {\em cherry reduction operations}. Since our type of phylogenetic networks may contain beads, we extend this convention
by collectively referring to a cherry reduction operation, a simplification of a type-1 degenerate cherry, and  
the trimming of a type-2 degenerate cherry as a {\em cherry modification operation}.

Following \cite{BESS21},  we call a sequence $\chi$ of elements in $X$ a
{\em complete cherry reduction sequence} for a beadless phylogenetic network
$N$ on $X$ if either (i) $\chi$ only contains $N$ if $N$ is a single vertex or (ii) $\chi$ is the sequence 
$N = N_0, N_1, N_2, \ldots, N_k, N_{k+1}$ of phylogenetic networks $N_i$, $0\leq i\leq k+1$, such that, for all $1\leq i\leq k+1$, the network 
$N_i$ is obtained from $N_{i-1}$ by a (single) cherry
reduction operation and $N_{k+1}$ is a single vertex. 
A (beadless) phylogenetic network that admits a  complete cherry reduction sequence is also called an {\em orchard}. 
With this in mind, we say that a phylogenetic network $N$ of our type has a {\em complete cherry modification sequence} if
$N$ has an augmented complete cherry reduction sequence 
in the sense that, in addition to cherry reduction operations, the only other permitted operation is 
simplification of a type-1 degenerate cherry.
For consistency reasons, we call a phylogenetic network 
that admits a complete cherry modification sequence also an {\em orchard} in this case.

 Similarly,
 we call a sequence of cherry modifications operations a {\em complete weak cherry modification sequence} for $N$,  if $N$ has an augmented complete cherry modification sequence 
 in the sense that, in addition to cherry reduction operations and the
 simplification of type-1 degenerate cherries, the trimming of a type-2 degenerate 
 cherry is also allowed. In that case, we also call $N$  
 a {\em weak orchard}.
 
For example and bearing in mind that the leaf labels are 
affected by the operations that govern the generation
of the simplification sequence for $\vec{m}$, 
the sequence of phylogenetic networks depicted
in Figure~\ref{fig:simpli-traceback-cherry-picking-only} read from right to left 
combined with the cherry modification sequence 
of the core network $\mathcal B(\vec{m})$ of $\vec{m}=(7,6,6,5)$
pictured in Figure~\ref{fig:gen-cherry-B(m)}
\begin{figure}[h]
	\centering
	\includegraphics[scale=0.4]{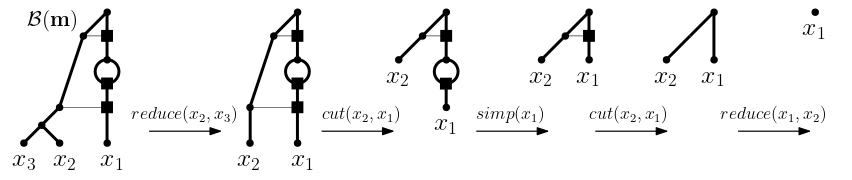}
	\caption{A complete cherry modification sequence for 
		the core network $\mathcal B(\vec{m})$ of $\vec{m}=(7,6,6,5)$. The applied cherry modifications operations are indicated above the arrows between the networks. 
		\label{fig:gen-cherry-B(m)}
	}
\end{figure}  
is a  complete cherry modification sequence for 
the realization $N(\vec{m})$ of $\vec{m}$ depicted in Figure~\ref{fig:simpli-traceback-cherry-picking-only}. On the
other hand, 
the sequence presented in Figure~\ref{fig:weak-orchard}
	\begin{figure}[h!]
	\centering
	\includegraphics[scale= 0.25]{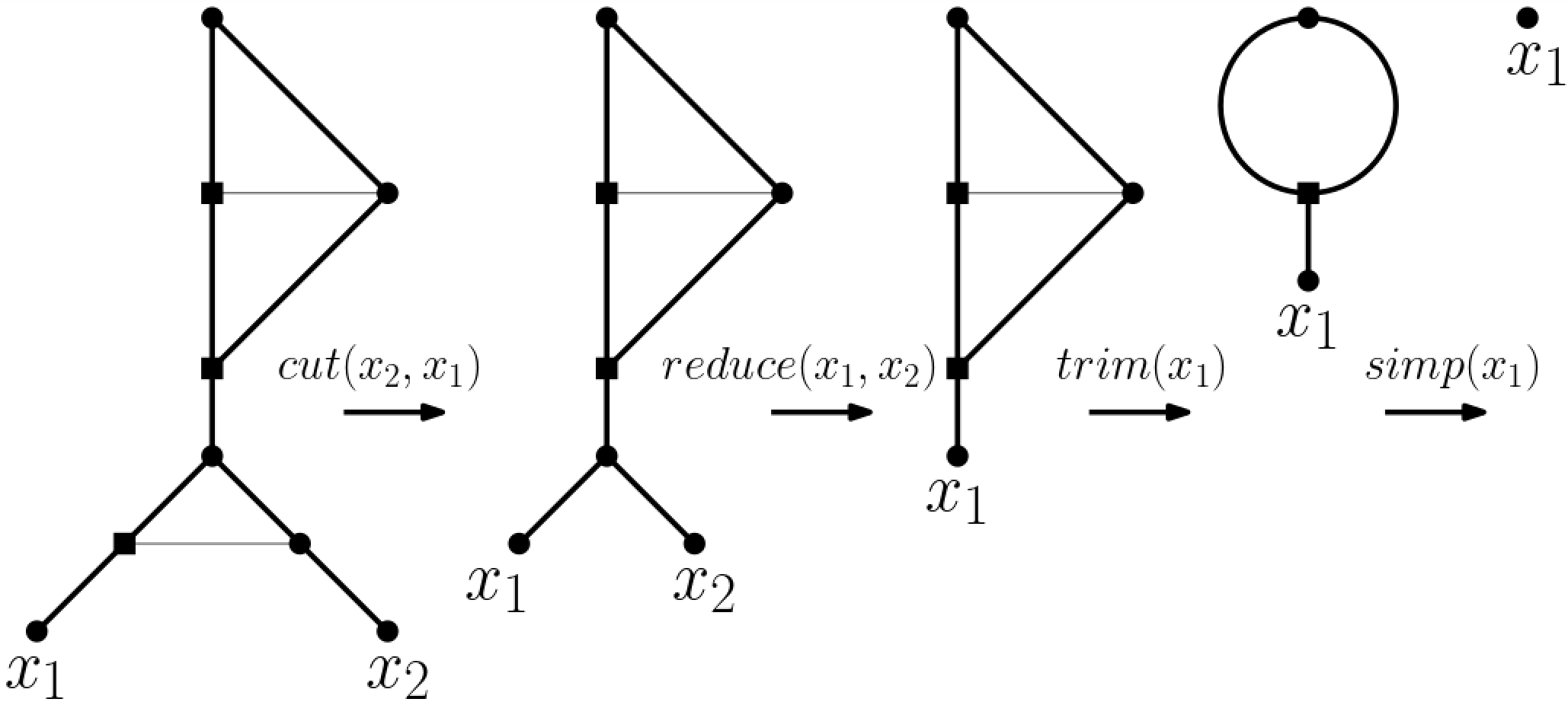}
	\caption{A weak cherry modification sequence for the \orange{realization $N(\vec{m})$, depicted on the left, of the} ploidy profile $\vec{m}=(6,3)$ on $\{x_1,x_2\}$. 
	\label{fig:weak-orchard}}
\end{figure}
is a weak cherry modification sequence for
$N(\vec{m})$ where $\vec{m}$ is the ploidy profile $(6,3)$ on $\{x_1,x_2\}$.
%$cut(x_2,x_1)$, 
%$reduce(x_1,x_2)$, $trim(x_1)$, $simp(x_1)$ induces a complete weak cherry modification sequence for
%$N(\vec{m})$ where $\vec{m}$ is the ploidy profile $(6,3)$ on $\{x_1,x_2\}$ (see also Section~\ref{sec:weak-orchard} in the appendix).

Note that neither a complete cherry reduction sequence 
nor a complete weak cherry modification sequence might 
exist for \orange{a realization of} a ploidy profile  and also that, in case it does
exist, \orange{such a realization} might have more than one.

 The fact that an orchard and also a weak orchard 
 induces a ploidy profile by
 taking numbers of directed paths from the root of the network
 to each of its leaves lies at the heart of our extension of 
 these concepts to ploidy profiles. More precisely, if $\vec{m}$ is 
a ploidy profile that is realized by a phylogenetic 
network $N$ and $N$ is an orchard then we call
 $\vec{m}$ an {\em orchard (with respect to $N$)}. 
 If $N$ is a weak orchard
 then we call $\vec{m}$  a {\em weak orchard (with respect to $N$)}. 
  To simplify terminology, we refer to $\vec{m}$ simply as an orchard or a weak orchard 
  if the knowledge of $N$ is of no relevance to the discussion.
    For example and each time initializing the construction of the realization $N(\vec{m})$ by 
    $\mathcal B(\vec{m})$, the ploidy profile $\vec{m}=(7,6,6,5)$ is an  orchard with respect to $N(\vec{m})$, and the ploidy profile
    $\vec{m}=(6,3)$ is a weak orchard with respect to
    its realization $N(\vec{m})$. Thus, $\vec{m}=(7,6,6,5)$ is an orchard and $\vec{m}=(6,3)$
    is a weak orchard. Furthermore,
    \orange{an exhaustive search for the ploidy profile $\vec{m}=(3)$ shows that there
    exist ploidy profiles that are a weak orchard but not an orchard.}
      
    The next result formalizes a link suggested by these two examples between complete cherry 
    modification sequences and simplification sequences. At its heart lies a 
    characterization of (beadless) orchards in terms of HGT-consistent labellings \cite[Theorem 1]{IJJM21}.

%\begin{figure}[h]
%	\centering
%	\label{fig:cherrypick}
%	\includegraphics[scale=0.4]{cherry-picking-only}
%	\caption{A cherry reduction sequence for the phylogenetic network $N^r(\vec{m})$ on $X = (a,b,c,d,e)$ which attains the ploidy profile $\vec{m} = (8,7,5,2,1)$.
%	}
%\end{figure}

%\begin{theorem}\cite[Theorem 2]{HM21}
%	\label{theo:main-paper1}
%	Suppose $\vec{m}$ is a ploidy profile on $X$ 
%with simplification sequence  $\vec{m}_0=\vec{m},\vec{m}_1,\vec{m}_2,\ldots, \vec{m}_l$
%		where, for all $0\leq i\leq l$, $\vec{m}^i=(m_i^1,\ldots, m_i^{p_i})$, some $p_i\geq 0$. If $m_i^1-m_i^2\leq m_i^2$
%		holds for
%		all $0\leq i\leq l$ then $N(\vec{m})$ attains
%		$\vec{m}$.
%\end{theorem} 

\begin{theorem}
	\label{theo:specialcase}
	Suppose $\vec{m}$ is a ploidy profile on $X$. If 
	$\vec{m}$ is practical
	%$\mathcal B(\vec{m})$ admits a HGT-consistent
	%labelling then 
	%the network $N(\vec{m})$ generated by
	%Algorithm~\ref{alg:construction} when given $\vec{m}$ and
	%$B$ as input is orchard. In particular,
	  then every ploidy profile in  the simplification sequence $\sigma(\vec{m})$ of $\vec{m}$ is  orchard. 
	  Furthermore, the  traceback through $\sigma(\vec{m})$ combined with
	 a cherry modification sequence for $\mathcal B(\vec{m})$ gives rise to a
	complete cherry modification sequence for $N(\vec{m})$
	provided the construction of $N(\vec{m})$ is initialized with $\mathcal B(\vec{m})$.	
	\end{theorem}
	\begin{proof}
		Since $\vec{m}$ is practical, Theorem~\ref{theo:arc-rich}
		implies that $\mathcal B(\vec{m})$ admits a HGT-consistent
		labelling. Combined with a canonical extension of
		 \cite[Theorem 1]{IJJM21} to our types of phylogenetic networks, it follows  that 
		there exists a complete cherry modification sequence for $\mathcal B(\vec{m})$. To see that 
		$N(\vec{m})$ has a complete cherry
		modification sequence it therefore suffices to show that at each step in the traceback of $\sigma(\vec{m})$ 
		only a cherry or a reticulate cherry is introduced.
		
		\orange{Assume for the remainder that the construction of $N(\vec{m})$ is initialized with $\mathcal B(\vec{m})$. Then 
			$N(\vec{m_t})$ has a complete cherry modification sequence by assumption on $\mathcal B(\vec{m})$ as $N(\vec{m_t})$  is $\mathcal B(\vec{m})$.}
		Using the same
		notation as in the construction of $N(\vec{m})$ outlined in Section~\ref{sec:realization} either $\alpha= 0$, or $\alpha>m_2'$, or $\alpha\leq m_2'$. Let $N''$ denote a realization for
		$\vec{m''}$ constructed from $\mathcal B(\vec{m})$ as described in the construction of $N(\vec{m})$. 

Employing the same indexing scheme as in the construction of 
$N(\vec{m})$, it follows that to
realize $\vec{m'}$, the leaf indexing the first component of $\vec{m''}$ is replaced by the cherry $\{x_1',x_2'\}$ if $\alpha =0$.
In the two remaining cases a single reticulate cherry on $\{x_1',x_2'\}$ with reticulate leaf $x_1'$ is generated. Thus, 
the generated realization of $\vec{m'}$, i.\,e.\, $N(\vec{m'})$, is orchard. It follows that
		every ploidy profile in $\sigma(\vec{m})$ is orchard.
		 The remainder of the theorem is an immediate
		consequence.
		\hfill{$\Box$}

		\end{proof}

Since, as mentioned in the proof of Theorem~\ref{theo:specialcase}, the reversal of the operations to construct the network
$N(\vec{m})$ from the core network $\mathcal B(\vec{m})$ corresponds to applying a single cherry reduction operation \orange{in each step of the traceback through $\sigma(\vec{m})$},
the companion result for ploidy profiles 
where $\mathcal B(\vec{m})$ admits a weak HGT-consistent labelling also holds.
Put differently, the result stated in Theorem~\ref{theo:specialcase} with the text ``If $\vec{m}$ is practical'' omitted, the word ``orchard'', replaced by ``weak orchard'' and the text ``concatenated with
a cherry modification sequence for $\mathcal B(\vec{m})$ results in a
complete cherry modification sequence for $N(\vec{m})$'' replaced by ``concatenated with
a weak cherry modification sequence for $\mathcal B(\vec{m})$ results in a
complete weak cherry modification sequence for $N(\vec{m})$'' also holds.

Intriguingly, the core network $\mathcal B(\vec{m})$ for $\vec{m}=(77,1,1,1)$ depicted in Figure~\ref{fig:core-net-77-1-1-1}(i) gives rise to a phylogenetic tree on $\{x_1,\ldots,x_4\}$ by deleting all horizontal arcs and removing one arc from each bead (each time suppressing the resulting vertices of indegree and outdegree one and the root in case this has rendered it a vertex with outdegree one). Beadless phylogenetic networks that enjoy this property are called tree-based \cite{FS15} and have recently attracted a considerable amount 
	of attention in the phylogenetic networks community (see, for example, \cite[Chapter 10.4.2]{S16}) since they can be thought of as phylogenetic trees to which arcs have been added. More precisely, a phylogenetic network $N$
	is called {\em tree-based} if there exists a phylogenetic tree $T$
	such that when first adding an incoming arc to the 
	root of  $T$ to obtain a tree $T'$
	and then subdividing some of the arcs of $T'$ and adding arcs between the generated subdivision vertices (ensuring
	that no directed cycle is created and the overall
	degree sum of the subdivision vertices is 3) the resulting
	directed graph is $N$. In that case, $T$ is called a
	{\em base tree} for $N$.
	
	As it turns out, the above
	observation for $\vec{m}=(77,1,1,1)$ and $\mathcal B(\vec{m})$ is not a coincidence as the following more general result holds. 
	%(see also \cite[Observation 1]{IJJM21} for the case of a beadless binary phylogenetic network). 

\begin{theorem}
	\label{theo:treebased}
	Suppose $\vec{m}$ is a ploidy profile on $X$. If $\vec{m}$ is practical then the network $N(\vec{m})$ generated from $\mathcal B(\vec{m})$ 
	is tree-based.
\end{theorem}
\begin{proof} This is an immediate consequence of 
	Theorem~\ref{theo:specialcase} 
	%which states that a simplification sequence for $\vec{m}$ 
	%gives rise to a complete cherry modification sequence 
	and the fact that the added horizontal arcs of $N(\vec{m})$
	correspond to reticulation arcs in reticulate cherries.
	\hfill{$\Box$}

	\end{proof}

Interestingly, the corresponding result for arc-rich ploidy profiles does not hold as the core network depicted in Figure~\ref{fig:core-net-77-1-1-1}(i) shows. 
%
%\documentclass[options]{class}
%\begin{figure}[h]
%	\centering
%	\label{fig:cherrypick}
%	\includegraphics[scale=0.2]{cherrypick}
%	\caption{The phylogenetic network $N(\vec{m})$ on $X = (a,b,c,d,e)$ attains the ploidy profile $\vec{m} = (8,7,5,2,1)$ 
%		is tree child and orchard but $\vec{m}$
%		does not satisfy the assumptions on
%		in Theorem~\ref{thm:specialcases} 
%		on $\alpha_{\vec{m}}(i)$, $0\leq i\leq |\sigma(\vec{m})|-1$.
%	}
%\end{figure}
%

%{\bf
%	If $\vec{m}$ is of the form $(n,n-1,n-2,\ldots, 1)$ then, for any ploidy profile obtained from $\vec{m}$ by removing at most one of its components other than the last but one, then
%	there exist a core network that is tree based.
%}

\section{A Viola \change{dataset}}
\label{sec:example}
In this section, we apply our findings to a simplified version
of a dataset studied in \cite{MJDBBBO12} to better understand the
evolutionary past of plants in the genus Viola.
The findings of the authors of that paper include a most
parsimonious PADRE reconstruction of allopolyploid relationships within
Viola, showing nine polyploidisation events (two of which involve more than two ancestral species) to explain the dataset's ploidy levels which range from
$2\times$ to $18\times$ \cite[Figure 4]{MJDBBBO12}. 
%
%
%%%%%%%%%%%%%% begin new text %%%%%%%%%%%%%%%%%%
%
%
To help ensure readability,  we present a simplified version of 
that network in Figure \ref{fig:viola}(i). To obtain it, we focused on (i) retaining the polyploidization events suggested by \cite[Figure 4]{MJDBBBO12} and the directed paths in the network which involve them, and (ii) representing subtrees
in terms of single leaves. More precisely, we removed the taxa: {\em V.diffusa, V.papuana, V.selkirkii, V.somchetica, V.tuberifa, V.renifola, V.principis,
V.vaginata, V.epipsila, V.pallens, V.lanceolata, V.primulifolia, V.jalapa{\"e}nsis,
V.occidentalis, V.pedata, V.clauseniana, V.sagittata, V.pubescens, V.lobata}. Furthermore, we summarized the taxa {\em V.capillaris} and {\em V.rubella}
into the label rubellium as they formed a cherry and the taxa {\em V.laricicola, V.striata, V.stagnina, V.uliginosa, V.mirabilis, V.chelmea, V.collina, V.hirta} into the label viola as they  formed a subtree. Finally, since the network in \cite[Figure 4]{MJDBBBO12} contained two vertices 
with indegree three, we have resolved them as indicated in Figure~\ref{fig:viola}(i). More precisely, the resolved vertices are the vertex 
labelled $10\times $ and its parent labelled 
$8\times$ and also the vertex labelled $18\times $ and its parent labelled 
$14\times $.
%%%%%%%%%%%%%%% end new text %%%%%%%%%%%%%%% 
 %
 %
%removing {\bf name them}, focusing on the polyploidization events and directed paths which involve them. 
 
\begin{figure}[h!]
	\centering
	\includegraphics[scale= 0.2]{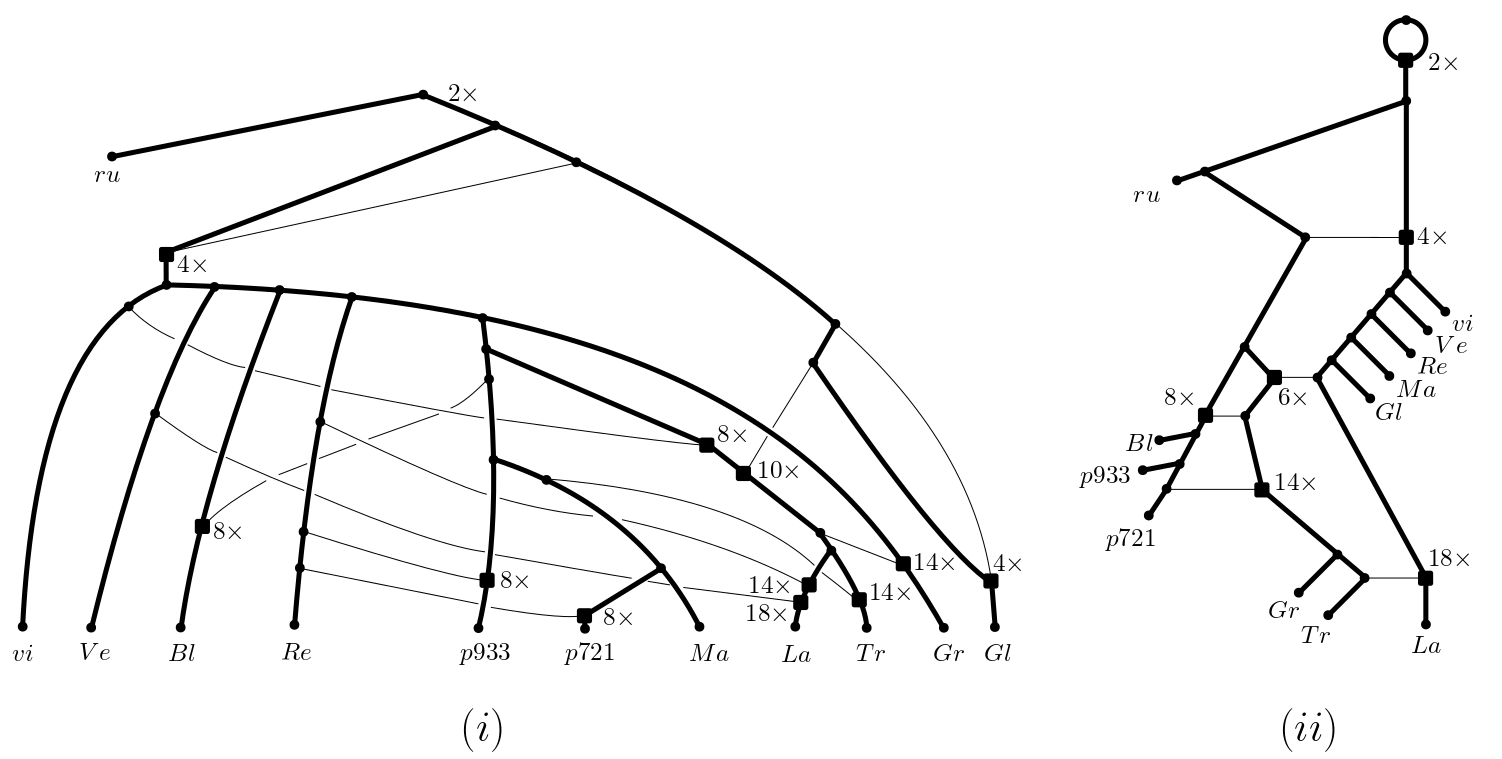}
	\caption{(i) A phylogenetic network on $X = \{$rubellium, viola,
		{\em V.verecunda}, {\em V.blande}, {\em V.repens}, {\em V.933palustris}, {\em V.721palustris}, {\em V.macloskeyi},
		{\em V.langsdorff}, {\em V.tracheliffolia}, {\em V.grahamii}, {\em V.glabella}$\}$ adapted from \cite{MJDBBBO12}. To improve clarity we include
		the ploidy level of each reticulation vertex. Apart from rubellium and viola which are denoted {\em ru} and {\em vi}, respectively,
		leaves are labelled by the first two
		characters of their name (omitting ``{\em V.}"). (ii) The realization $N(\vec{m})$
		of the ploidy profile $\vec{m}$ induced
		by the network in (i). Contrary to the network in (i),
		$N(\vec{m})$ is orchard. In each case, the graph obtained by removing the non-bold arcs is a base tree for $N(\vec{m})$. For ease of readability, the labels of the non-leaf vertices represent the number of directed paths from the root to that vertex in each of (i) and (ii).
		\label{fig:viola}
	}
\end{figure}

Although the network pictured in Figure~\ref{fig:viola}(i) clearly represents the 
ploidy profile $\vec{m} = (18,14,14,10,8,8,8,4,4,2)$
by taking the number of directed paths from the root to each leaf,
from  a formal point of view, it is not a realization of 
$\vec{m}$  since the ancestral species at the root  is assumed to be diploid. This shortcoming of our framework can however easily be 
rectified by adding a bead via an incoming arc to the root of the network.

As was established in \cite[Theorem 2]{HM21}, the minimum number of 
reticulation vertices required
by a phylogenetic network to realize $\vec{m}$ is 5. Since 
the phylogenetic network $N(\vec{m})$ pictured in Figure~\ref{fig:viola}(ii)
realizes $\vec{m}$ using five reticulation vertices
it follows that it is optimal with regards to this property.
Furthermore since none of the five reticulation vertices 
are contained in a bead, they all represent allopolyploidization
events. Finally, $N(\vec{m})$
admits a HGT-consistent labelling which implies that
$\vec{m}$ is orchard. In turn, this implies that 
a phylogenetic network that  realizes $\vec{m}$ can be obtained from a phylogenetic tree (in this case without beads) 
by adding horizontal arcs. Given these attractive features it could be of interest to better understand to what
extent the network $N(\vec{m})$ can be used to inform the construction of a multiple-labelled tree 
such as the one underpinning the network in Figure~\ref{fig:viola}(i). As mentioned above already,
constructing such a tree
is not easy in general \cite{HLMS09}.

Using the insights from Section~\ref{sec:comparison}
to help assess how different the two networks in
Figure~\ref{fig:viola} are, assume that the chosen
distance measure for comparing two multiple-labelled trees is the SPR-distance.
Then by first applying split-operations to each of 
the two networks pictured in Figure~\ref{fig:viola} until
a multiple-labelled tree is obtained and then transforming 
one of the two obtained multiple-labelled trees into the other
via a sequence $\chi_{SPR}$ of multiple-labelled trees such that
two consecutive multiple-labelled trees in $\chi_{SPR}$ have SPR-distance 1 yields an upper bound of 26 on the  $D_{\mathcal P(\vec{m})}$-distance between the two networks. 
%As is easy to
%verify and independent of the tree-editing operation used
%to transform a multiple-labelled tree into another, the $D_{\mathcal P(\vec{m})}$-distance is in fact always a metric
%i.e.  $D_{\mathcal P(\vec{m})}$ does not attain negative values, is zero precisely if $N$ equals $N'$, is symmetric, and satisfies the triangle inequality. 
\bigskip

\section{Concluding remarks}

In this paper, we have pushed back the current limits of the emerging field of 
{\em Polyploid Phylogenetics} \cite{R21}
by studying combinatorial properties of a ploidy profile $\vec{m}$ on some set
$X$. Denoting by $N(\vec{m})$ the phylogenetic network obtained as a slightly modified version of the construction of a phylogenetic network that appeared in \cite{HM21}, we show that $N(\vec{m})$
may be viewed as a generator of ploidy profile space $\mathcal P(\vec{m})$
in the sense that any other realization $N$  of $\vec{m}$ can be obtained from it by going along the edges of a path from $N(\vec{m})$ to $N$ in $\mathcal P(\vec{m})$
(Proposition~\ref{prop:connected}). Furthermore, $N(\vec{m})$  may be thought of as a phylogenetic tree with beads  to which additional arcs have been added (Theorem~\ref{theo:treebased}) and at most one of these additional arcs cannot be drawn as a horizontal arc
(Theorem~\ref{theo:arc-rich}). 
Furthermore, we establish a close link between the concept of a cherry reduction sequence for  $N(\vec{m})$
and the simpification sequence for $\vec{m}$, a concept which underpins the construction of $N(\vec{m})$   (Theorem~\ref{theo:specialcase}). As an immediate consequence, we also have that the ploidy profile space for the ploidy profiles described in Corollary~\ref{cor:special-case} can be generated from a phylogenetic tree with beads and only horizontal arcs added\change{.}  Finally, we illustrate our findings by means of a real biological dataset.

Although our results are encouraging, numerous open questions remain. From a more biological point of view, these include understanding how well the
$D_{\mathcal P(\vec{m})}$-distance captures similarity between different realizations of a ploidy profile
$\vec{m}$. 
%Simulation studies are the obvious choice to address this questions.
In the context of this it should be noted that the edit-distance type nature of 
the $D_{\mathcal P(\vec{m})}$-distance implies that, in general,
it might be computationally difficult to compute it. This immediately begs the more mathematical question of how to bound it. 
%An implementation of the $D_{\mathcal P(\vec{m})}$-distance will undoubtedly be beneficial for making a first step into this direction. 

Also, it might be
useful to explore diameter bounds for the $D_{\mathcal P(\vec{m})}$-distance and the effect the choice of
distance measure on multiple-labelled trees has. The same
also holds when replacing the sequence of split operations
to obtain a multiple-labelled tree from a phylogenetic network
with the ``unfold'' operation for phylogenetic networks.
Essentially, this operation associates a multiple-labelled tree 
to a phylogenetic network $N$ by recording for every
leaf $x$ of $N$ all directed paths from the root of $N$
to $x$ (see e.\,g.\,\cite{HM06,HOLM06} for details about this operation). \change{It may also be interesting to explore the relationship between the simplification sequence and trinets \cite{HM13}. For example, how are the classes of phylogenetic networks that can be determined from trinets related to the class of phylogenetic networks with complete cherry reduction sequences?}

In a different direction, it might be of interest to see if the 
results and approaches presented here can be extended to include further evolutionary processes
such as aneuploidy whereby only a subset of the chromosome set of a genome (as opposed to the whole set of chromosomes) occurs multiple times. This could potentially involve representing a polyploid species not in terms of a ploidy level but in terms of a vector with each component representing the number of times the chromosome indexing it is found. A ploidy profile would in that case not be a vector of positive integers but a vector of vectors, each of them indexed by a species. Although attractive at first glance, this would require finding a new way of realizing a ploidy profile in terms of a phylogenetic network. 

From a more mathematical point of view, it might also be interesting to investigate if a ploidy profile $\vec{m}$ whose simplification sequences terminates in a practical ploidy profile can be characterized without having to compute the  simplification sequence of 
$\vec{m}$ first.
This might require a better understanding of the link between
simplifications sequences and  ploidy profiles that are orchard. The availability of such a characterization could potentially lead to a fast way to decide if a ploidy profile $\vec{m}$ can be realized by an
orchard.

In the case of beadless phylogenetic networks, relationships between various types of properties are known. For example
every orchard is also tree-based \cite[Corollary~2]{IJJM21}
and also every tree-child network is orchard \cite{BS16}.
Tree-child networks are defined as those (beadless) phylogenetic networks for which, 
for every one of its vertices $v$, there exists a directed path $P$ to a leaf so that no vertex on $P$
other than potentially $v$ is a reticulation vertex. Extending this property canonically to our types of
networks by allowing $P$ to contain reticulation vertices in beads results in a natural way to extend the tree-child concept to our types of phylogenetic networks. More precisely, we call a ploidy profile {\em tree-child} if it has a realization that is tree-child when
reticulation vertices in beads are ignored. 

As is easy to see, if the construction of 
$N(\vec{m})$ is initialized with the core network $\mathcal B(\vec{m})$ then any ploidy profile of the form 
$(2^n,2^{n-1},\ldots,2^1,2^0)$ with $n\geq 1$ is tree-child. However at the same time \orange{an exhaustive search for} the ploidy profile
$\vec{m}=(3)$ shows that not every ploidy profile is tree-child.  It might therefore be interesting
to characterize tree-child ploidy profiles. This might involve better understanding properties of the core network for $\vec{m}$
with which the construction of $N(\vec{m})$ is initialized
(see Figure~\ref{fig:twocores12} for two alternative choices
of a core network of the ploidy profile $(12,1)$ one of
which is $\mathcal B(\vec{m})$ and the other is not of the form $\mathcal B(\vec{m})$).
 As part of this it might be tempting to first focus on core networks obtained from a prime factor decomposition of the single component of a strictly simple ploidy profile (see also  \cite{HM21} for more on this). 

\section*{Acknowledgement}
Both  authors would like to thank the two anonymous reviewers for their helpful comments to improve an earlier version of the paper.

\section*{Data availability statement}
Apart from data already publicly available (see \cite{MJDBBBO12}), the manuscript has no data 
associated to it.

 \bibliographystyle{plain}
 \bibliography{bibliography2022-02-09.bib}
 
 \end{document}